\documentclass[11pt]{article}
\usepackage[T1]{fontenc}
\usepackage[utf8]{inputenc}
\usepackage{tgpagella}
\usepackage{amssymb}
\usepackage{amsmath}
\usepackage[top=1in, bottom=1in, left=1in, right=1in, letterpaper]{geometry}
\usepackage{setspace}
\usepackage{amsfonts}
\usepackage{amstext}
\usepackage{amsthm}
\usepackage[round]{natbib}
\usepackage{lscape}
\usepackage{threeparttable}
\usepackage{hyperref}
\usepackage{graphicx}
\usepackage{url}
\usepackage{color}
\usepackage{comment}

\numberwithin{equation}{section}
\numberwithin{table}{section}
\numberwithin{figure}{section}
\newtheorem{assumption}{Assumption}
\newtheorem{lemma}{Lemma}
\newtheorem{theorem}{Theorem}
\newtheorem{corollary}{Corollary}
\newtheorem{proposition}{Proposition}
\newtheorem{algorithm}{Algorithm}
\theoremstyle{remark}
\newtheorem{remark}{Remark}
\input{tcilatex}
\begin{document}

\title{Sparse Quantile Regression}
\author{Le-Yu Chen\thanks{%
E-mail: lychen@econ.sinica.edu.tw} \\
{\small {Institute of Economics, Academia Sinica}} \and Sokbae Lee\thanks{%
E-mail: sl3841@columbia.edu} \\
{\small {Department of Economics, Columbia University}}\\
{\small {Centre for Microdata Methods and Practice, Institute for Fiscal
Studies} }}
\date{March 5, 2023}
\maketitle

\begin{abstract}
We consider both $\ell _{0}$-penalized and $\ell _{0}$-constrained quantile
regression estimators. For the $\ell _{0}$-penalized estimator, we derive an
exponential inequality on the tail probability of excess quantile prediction
risk and apply it to obtain non-asymptotic upper bounds on the mean-square
parameter and regression function estimation errors. We also derive
analogous results for the $\ell _{0}$-constrained estimator. The resulting
rates of convergence are nearly minimax-optimal and the same as those for $%
\ell _{1} $-penalized and non-convex penalized estimators.
Further, we characterize expected Hamming loss for the $\ell _{0}$-penalized
estimator. We implement the proposed procedure via mixed integer linear
programming and also a more scalable first-order approximation algorithm. We
illustrate the finite-sample performance of our approach in Monte Carlo
experiments and its usefulness in a real data application concerning
conformal prediction of infant birth weights (with $n\approx 10^{3}$ and up
to $p>10^{3}$). In sum, our $\ell _{0} $-based method produces a much
sparser estimator than the $\ell _{1}$-penalized and non-convex
penalized approaches without compromising precision. \newline



\end{abstract}

\newpage

\onehalfspacing

\section{Introduction}

Quantile regression has been increasingly popular since the seminal work of %
\citet{Koenker1978}. See \citet{Koenker2005} for a classic and comprehensive
text on quantile regression and \citet{Koenker17} for a review of recent
developments. This paper is concerned with estimating a sparse
high-dimensional quantile regression model: 
\begin{equation}
Y=X^{\top }\theta _{\ast }+U,  \label{the linear QR model}
\end{equation}%
where $Y\in \mathbb{R}$ is the outcome of interest, $X\in \mathbb{R}^{p}$ is
a $p$-dimensional vector of covariates, $\theta_\ast$ is the vector of
unknown parameters, and $U$ is a regression error. Let $Q_{\tau }(U|X)$
denote the $\tau $-th quantile of $U$ conditional on $X$. Assume that $%
Q_{\tau }(U|X)=0$ almost surely for a given $\tau \in (0,1)$ and that the
data consist of a random sample of $n$ observations $\left(
Y_{i},X_{i}\right) _{i=1}^{n}$. As usual, $p$ can be much larger than $n$;
however, sparsity $s$, the number of nonzero elements of $\theta_\ast$, is
less than $n$.

To date, an $\ell _{1}$-penalized approach to estimating 
\eqref{the linear
QR model} has been predominant in the literature mainly thanks to its
computational advantages. See e.g., \citet{belloni2011}, \citet{Wang13}, %
\citet{belloni2014uniform, Belloni2019}, \citet{zheng2015}, %
\citet{Lee2018oracle}, \citet{
Lv2018}, \citet{Wang18wild} and \citet{wang2019wp} among many others. The $%
\ell _{1}$-penalized quantile regression ($\ell _{1}$-PQR hereafter) is akin
to the well known approach of Lasso \citep{tibshirani1996}. Smooth yet
non-convex penalized estimation approaches have also been proposed as
alternatives to $\ell _{1}$-PQR. These include methods of adaptive Lasso (adaptive $\ell
_{1}$-) and non-convex penalized quantile regressions %
\citep[see e.g.,][]{Wu:Liu:09,Wang12,Fan14adaptive,Fan2014,Peng:Wang:15}. See also \citet{wang2022} for the
state-of-the-art theoretical analysis of $\ell _{1}$-based and non-convex
penalized quantile regressions.

Recently, there is emerging interest in adopting an $\ell _{0}$-based
approach since the latter is regarded as a more direct solution to
estimation problem under sparsity. For instance, \citet{bertsimas2016} took
an $\ell _{0}$-constrained approach in order to solve the best subset
selection problem in linear regression models. \citet{Huang:2018} proposed a
scalable computational algorithm for $\ell _{0} $-penalized least squares
solutions. \citet{chen2018, chen2018arXiv} studied the $\ell _{0}$%
-constrained and $\ell _{0}$-penalized empirical risk minimization
approaches to high dimensional binary classification problems. %
\citet{Bertsimas:Parys:2020} and \citet{Hazimeh:Mazumder:2020} made further
advances in $\ell _{0}$-based methods for mean regression models. %
\citet{DHM:jmlr:2021} considered an $\ell _{0}$-regularized approach for a
hinge-loss based classification problem.

In this paper, we pursue an $\ell _{0}$-based approach to estimating sparse
quantile regression. We are inspired by \citet[Section
6]{bertsimas2016}, who provided a piece of numerical evidence---without
theoretical analysis---that the $\ell _{0}$-constrained least absolute
deviation (LAD) estimator outperforms $\ell _{1}$-penalized LAD estimator in
terms of both sparsity and predictive accuracy. That is, %
\citet{bertsimas2016} made a convincing case for adopting an $\ell _{0}$%
-based approach in median regression. In convex optimization, a constrained
approach is equivalent to a penalized method 
\citep[see,
e.g.,][]{boyd2004convex}. For non-convex problems, both are distinct and it
is unclear which method is better. Therefore, in the paper, we consider both 
$\ell _{0}$-constrained and $\ell_{0}$-penalized quantile regression ($%
\ell_0 $-CQR and $\ell_0$-PQR hereafter).

The main contributions of this paper are twofold. First, we derive an
exponential inequality on the tail probability of the excess quantile
predictive risk and apply it to obtain non-asymptotic upper bounds on a
triplet of population quantities for the $\ell _{0}$-PQR estimator: the mean
excess predictive risk, the mean-square regression function estimation
error, and the mean-square parameter estimation error. The resulting rates
of convergence for the triplets are at the order of $s \ln p /n$, which is
the same as those of $\ell_1$-PQR and non-convex penalized quantile regression 
\citep[see,
e.g.,][]{belloni2011,wang2019wp,wang2022} and nearly matches the minimax
lower bound $s \ln (p/s) /n$ obtained in Theorem 4.1 of \citet{wang2019wp}.
However, the optimal tuning parameter $\lambda$ in $\ell_1$-PQR is of order $%
\sqrt{\ln p/n}$, whereas it is of order ${\ln p/n}$ in $\ell_0$-PQR. We also
characterize expected Hamming loss for the $\ell_0$-penalized estimator. In
a nutshell, $\ell_0$-PQR produces a sparser estimator than $\ell_1$-PQR,
while maintaining the same level of prediction and estimation errors. In
addition, we establish analogous results for the $\ell_0$-CQR estimator
under the assumption that the imposed sparsity is at least as large as true
sparsity. Our non-asymptotic results build on \citet{bousquet2002} and %
\citet{massart2006} and are applicable for $\ell_0$-based, general $M$%
-estimation with a Lipschitz objective function that includes sparse
logistic regression as a special case. Therefore, our theoretical results
may be of independent interest beyond quantile regression.

The second contribution is computational. Both $\ell_0$-CQR and $\ell_0$-PQR
estimation problems can be equivalently reformulated as mixed integer linear
programming (MILP) problems. This reformulation enables us to employ
efficient mixed integer optimization (MIO) solvers to compute exact
solutions to the $\ell _{0}$-based quantile regression problems. 
However, the method of MIO is concerned with optimization over integers,
which could be computationally challenging for large scale problems. To
scale up $\ell _{0}$-based methods, \citet{bertsimas2016} developed fast
first-order approximation methods for both $\ell _{0}$-constrained least
squares and absolute deviation estimators. \citet{Huang:2018} also proposed
a scalable computational algorithm for approximating the $\ell _{0}$%
-penalized least squares solutions. Building on these papers, we propose a
new first-order computational approach, which can deliver high-quality approximate $\ell _{0}$-PQR solutions and thus can be used as a warm-start
strategy for boosting the computational performance of the MILP based implementation approach. As a
standalone algorithm, our first-order approach renders the $\ell _{0}$-PQR
computationally as scalable as commonly used $\ell_1$-PQR.

As an illustrative application, we consider conformal prediction of birth
weights and have a horse race among $\ell _{0}$-CQR, $\ell _{0}$-PQR, $\ell
_{1}$-PQR, adaptive Lasso and non-convex penalized quantile regressions with $n\approx
1000$ and $p$ ranging from $p\approx 20$ to $p\approx 1600$. Recently, %
\citet{romano2019conformalized} combined conformal prediction with quantile
regression and proposed conformalized quantile regression that rigorously
ensures a non-asymptotic, distribution-free coverage guarantee, independent
of the underlying regression algorithm. When we implement conformal
prediction using competing estimation methods, we find that both $\ell _{0}$%
-CQR and $\ell _{0}$-PQR are capable of delivering much sparser solutions
than $\ell _{1}$-PQR, while maintaining tighter width yet comparable coverage of prediction confidence intervals in the high dimensional settings. Adaptive Lasso and non-convex penalized estimation approaches do improve the performance of $\ell _{1}$-PQR, but they also tend to select more covariates than $\ell _{0}$-CQR and $\ell _{0}$-PQR and perform poorly especially when the covariates are highly dependent on each other.
Furthermore, we obtain similar results in Monte Carlo experiments.
Therefore, $\ell _{0}$-CQR and $\ell _{0}$-PQR are worthy competitors to $%
\ell _{1}$-PQR, adaptive Lasso and non-convex penalized quantile regression approaches
---superior if a researcher prefers sparsity---as supported by
non-asymptotic theory, a real-data application and Monte Carlo experiments.

The rest of this paper is organized as follows. In Section \ref%
{Sec:L0-penalized QR}, we set up the sparse quantile regression model and
present the $\ell _{0}$-based approaches. In Section \ref{Sec:Theory}, we
establish non-asymptotic statistical properties of the proposed $\ell _{0}$%
-PQR and $\ell _{0}$-CQR estimators. In Section \ref{Sec:Computation}, we
provide both MILP- and first-order (FO)-based computational approaches for
solving the $\ell _{0}$-PQR problems. In Section \ref{Sec:Simulation}, we
perform a simulation study on the finite-sample performance of our proposed
estimators. In Section \ref{Sec:Application}, we illustrate our method in a
real data application concerning conformal prediction of birth weights. We
then conclude the paper in Section \ref{Sec:Conclusions}. Appendix \ref%
{Proofs} collates proofs of all theoretical results of the paper and an online appendix contains further details on the variable selection results of our empirical study.

\section{\texorpdfstring{$\ell_0$}{L0}-Based Approaches to Quantile
Regression\label{Sec:L0-penalized QR}}

Let $\Vert a\Vert _{0}$ be the $\ell _{0}$ norm of a vector $a$, which is
the number of nonzero components of $a$. The usual $\ell _{1}$ and $\ell
_{2} $ norms are denoted by $\Vert \cdot \Vert _{1}$ and $\Vert \cdot \Vert
_{2}$, respectively. For any $t$ and $u$, let 
\begin{equation}
\rho (t,u)\equiv (t-u)[\tau -1(t\leq u)].  \label{quantile check function}
\end{equation}%
Let $\Theta $ denote a parameter space, which is assumed to be a compact
subspace of $\mathbb{R}^{p}$. Define%
\begin{equation}
S_{n}(\theta )\equiv n^{-1}\sum_{i=1}^{n}\rho (Y_{i},X_{i}^{\top }\theta ).
\label{Sn}
\end{equation}

We first define $\ell_0$-CQR. For any given sparsity $q\geq 0$, let $%
\widetilde{\theta }$ denote an $\ell _{0}$-constrained quantile regression ($%
\ell _{0}$-CQR) estimator, which is defined as a solution to the following
minimization problem: 
\begin{equation}
\min\limits_{\theta \in \mathbb{B}(q)}S_{n}(\theta ), \text{ where } \; 
\mathbb{B}(q)\equiv \{\theta \in \Theta :\Vert \theta \Vert _{0}\leq q\}.
\label{L0-CQR-est}
\end{equation}
In practice, choosing $q$ is important: $\ell_0$-CQR will result in
selecting many more (or far fewer) covariates if the imposed sparsity is too
large (or too small).

To mitigate the issue of unknown true sparsity $s$, we now focus on $\ell
_{0}$-PQR. Let $\widehat{\theta }$ denote an $\ell _{0}$-PQR estimator,
which is defined as a solution to the following minimization problem: 
\begin{equation}
\min\nolimits_{\theta \in \mathbb{B}(k_{0})}S_{n}(\theta )+\lambda \Vert
\theta \Vert _{0},  \label{penalized QR}
\end{equation}%
where $\lambda $ is a nonnegative tuning parameter and $k_{0}$ is a fixed
upper bound for the true sparsity $s$. In $\ell _{0}$-PQR, the main tuning
parameter is $\lambda $. To adapt to an unknown $s$, we rely on $\ell _{0}$%
-penalization that is controlled by $\lambda $. The sparsity bound $k_{0}$
is different from $q$ in $\ell _{0}$-CQR. The latter acts as a tuning
parameter, which will be calibrated to maximize the predictive performance,
whereas the former is predetermined and imposed throughout the implementation of $\ell_{0}$-PQR. We will set $k_{0}$ with a large value in numerical exercises. 

To make our proposed estimators operational, we follow the standard machine
learning approach. That is, we first randomly split the dataset into three
samples: training, validation and test samples. For each candidate value of
the tuning parameter $q$ or $\lambda$, we estimate the model using the
training sample. Then, the tuning parameter is selected based on the
quantile prediction risk using the validation sample. Finally, out-of-sample
performance is evaluated using the test sample.

\section{Theory for \texorpdfstring{$\ell _{0}$}{L0}-Based Quantile
Regression}

\label{Sec:Theory}

\subsection{Assumptions\label{theory:assump}}

We provide general regularity conditions that include quantile regression as
a special case. Define $S(\theta )\equiv \mathbb{E}\left[ \rho (Y,X^{\top
}\theta )\right] $.

\begin{assumption}
\label{iden-cond} $S(\theta )\geq S(\theta _{\ast })$ for any $\theta \in
\Theta $.
\end{assumption}

Note that for quantile regression, 
\begin{equation}
S(\theta )-S(\theta _{\ast })=\int \int_{0}^{x^{\top }(\theta -\theta _{\ast
})}\left[ F_{U|X}(z|x)-F_{U|X}(0|x)\right] dz\;dF_{X}(x),
\label{difference in QR objective functions}
\end{equation}%
where $F_{U|X}(\cdot |x)$ is the cumulative distribution function of $U$
conditional on $X=x$ and $F_{X}$ is the cumulative distribution function of $%
X$. Thus, Assumption \ref{iden-cond} is satisfied.

\begin{assumption}
\label{L-cond}There exists a Lipschitz constant $L$ such that 
\begin{equation}
\left\vert \rho (t,u_{1})-\rho (t,u_{2})\right\vert \leq L\left\vert
u_{1}-u_{2}\right\vert  \label{Lipschitz continuity}
\end{equation}%
for all $t,u_{1},u_{2}\in \mathbb{R}$.
\end{assumption}

Assumption \ref{L-cond} is satisfied for quantile regression with $L=1$. For
any two real numbers $x$ and $y$, let $x\vee y\equiv \max \{x,y\}$ and $%
x\wedge y\equiv \min \{x,y\}$.

\begin{assumption}
\label{design-cond-1}There exists a positive and finite constant $B$ such
that 
\begin{equation}
\max_{1\leq j\leq p}\{\left\vert X^{(j)}\right\vert \vee \left\vert \theta
^{(j)}\right\vert \}\leq B,  \label{boundedness}
\end{equation}%
where $X^{(j)}$ and $\theta ^{(j)}$ denote the $j$-th component of $X$ and
that of $\theta $, respectively.
\end{assumption}

Assumption \ref{design-cond-1} requires that each component of $X$ and that
of $\theta $ be bounded by a universal constant. This condition could be
restrictive yet is commonly adopted in the literature. For example, \cite%
{zheng2018} assumed the uniform boundedness of each of the covariates,
citing the literature that points out that \textquotedblleft a global linear
quantile regression model is most sensible when the covariates are confined
to a compact set\textquotedblright\ to avoid the problem of quantile
crossing.

\begin{assumption}[Separability Condition]
\label{sep-cond} There exists a countable subset $\Theta ^{\prime }$ of $%
\Theta $ that satisfies the following conditions: (i) for any $\theta \in
\Theta $, there exists a sequence $(\theta _{j})$ of elements of $\Theta
^{\prime }$ such that, for every realization of $(Y,X)$, $\rho (Y,X^{\top
}\theta _{j})$ converges to $\rho (Y,X^{\top }\theta )$ as $j\rightarrow
\infty $. (ii) Furthermore, for any given $\varepsilon _{\ast }>0$, there
exists a point $\theta _{\ast }^{\prime }\in \Theta ^{\prime }$ such that $%
\Vert \theta _{\ast }^{\prime }\Vert _{0}=\Vert \theta _{\ast }\Vert _{0}$
and $S(\theta _{\ast }^{\prime })\leq S(\theta _{\ast })+\varepsilon _{\ast
} $.
\end{assumption}

Assumption \ref{sep-cond} is very mild. A similar condition is assumed in %
\citet{massart2006} to avoid measurability issues and to use the
concentration inequality by \citet{bousquet2002}. By Assumption \ref{L-cond}
and taking $\Theta ^{\prime }=\Theta \cap \mathbb{Q}^{p}$, Assumption \ref%
{sep-cond} (i) holds by the denseness of the set of rational numbers and the
continuity of the function $\rho $. Suppose that, for some non-negative
random variable $Z$ with $\mathbb{E}\left( Z\right) <\infty $, and $\rho
(Y,X^{\top }\theta )\leq Z$ holds with probability 1 for every $\theta \in
\Theta $. Then using this condition together with Assumption \ref{sep-cond}
(i), we can also deduce from the dominated convergence theorem that $%
S(\theta _{\ast })=\inf_{\theta \in \Theta ^{\prime }}S(\theta )$ and thus
Assumption \ref{sep-cond} (ii) also holds. In the quantile regression case,
we can take the dominating variable $Z$ to be $\left\vert Y\right\vert
+pB^{2}$, which has finite mean provided that the mean of $|Y|$ is also
finite. The requirement that $\mathbb{E}|Y|<\infty $ is not strictly
necessary because we can redefine the quantile regression objective function
by $\rho (Y,X^{\top }\theta )-\rho (Y,X^{\top }\theta _{\ast })$, whose
magnitude is uniformly bounded above by $2pB^{2}$.

For each $\theta $, define $R(\theta )\equiv \mathbb{E}[|X^{\top }(\theta
-\theta _{\ast })|^{2}],$ which is the the expected squared difference of
the true quantile regression function $X^{\top }\theta _{\ast }$ and a
linear fit evaluated at a given parameter vector $\theta $.

\begin{assumption}
\label{key-cond}For some $k\geq k_{0}$ in \eqref{penalized QR}, there exists
a constant $\kappa _{0}>0$ such that 
\begin{equation}
S(\theta )-S(\theta _{\ast })\geq \kappa _{0}^{2}R(\theta )\text{ for all }%
\theta \in \mathbb{B}(k).  \label{margin condition}
\end{equation}
\end{assumption}

Assumption \ref{key-cond} relates $R(\theta )$ to the difference of their
corresponding quantile predictive risks. Given Assumption \ref{design-cond-1}%
, if, for some $k\geq k_{0}$, the distribution $F_{U|X}\left( z|x\right) $
admits a Lebesgue density $f_{U|X}\left( z|x\right) $ that is bounded below
by a positive constant $c_{u}$ for all $z$ in an open interval containing $%
\left[ -B^{2}\left( k+s\right) ,B^{2}\left( k+s\right) \right] $ and for
all $x$ in the support of $X$, then Assumption \ref{key-cond} holds with $%
\kappa _{0}=\sqrt{c_{u}/2}.$

\begin{assumption}
\label{design-cond-2} For some $k\geq k_{0}$ in \eqref{penalized QR}, there
exists a constant $\kappa _{1}>0$ such that 
\begin{equation}
R(\theta )\geq \kappa _{1}^{2}\left\Vert \theta -\theta _{\ast }\right\Vert
_{2}^{2}\text{ for all }\theta \in \mathbb{B}(k).  \label{sparse-eigen}
\end{equation}
\end{assumption}

For any subset $J\subset \{1,...,p\}$, let $X_{J}$ denote the $\left\vert
J\right\vert $-dimensional subvector of $X\equiv (X^{(1)},\ldots
,X^{(p)})^{\top }$ formed by keeping only those elements $X^{(j)}$ with $%
j\in J$. Suppose that, for some $k\geq k_{0}$ and for any subset $J\subset
\{1,...,p\}$ such that $\left\vert J\right\vert \leq k+s$, the smallest
eigenvalue of $\mathbb{E}\left( X_{J}X_{J}^{\top }\right) $ is bounded below
by a positive constant $\omega $. Since $R(\theta )=(\theta -\theta _{\ast
})^{\top }\mathbb{E}\left[ XX^{\top }\right] (\theta -\theta _{\ast })$ and $%
\Vert \theta -\theta _{\ast }\Vert _{0}\leq k+s$ for $\theta \in \mathbb{B}%
(k),$ it then follows that Assumption \ref{design-cond-2} holds with $\kappa
_{1}=\sqrt{\omega }$. This assumption is related to the sparse eigenvalue
condition used in the high dimensional regression literature (see, e.g. %
\citet{Raskutti2011}). For example, if $X$ is a random vector with mean zero
and the covariance matrix $\Sigma $ whose $(i,j)$ component is $\Sigma
_{i,j}=r^{|i-j|}$ for some constant $r>0$, then the smallest eigenvalue of $%
\Sigma $ is bounded away from zero where the lower bound is independent of
the dimension $p$ \citep[][p.\thinspace1384]{van2009} and thus $\mathbb{E}%
\left( X_{J}X_{J}^{\top }\right) $ is bounded below by a universal positive
constant for every $J\subset \{1,...,p\}$.

\subsection{Non-Asymptotic Bounds and Minimax Optimal Rates}

The following theorem is the key step to the main results of this section
for $\ell_0$-PQR.

\begin{theorem}
\label{thm-main-1}Let Assumptions \ref{iden-cond}--\ref{design-cond-2} hold.
Suppose that $s\leq k_{0}$. Then, for any given positive scalar $\eta \leq 1$%
, there is a universal constant $M$, which depends only on $\eta $, such
that, for every $y\geq 1$, 
\begin{align}
& \mathbb{P}\left[ S(\widehat{\theta })-S\left( \theta _{\ast }\right) \geq 
\frac{2\lambda s}{1-\eta }+32C^{2}(s+k_{0})\left( \frac{1+\eta +M\eta y}{%
1-\eta }\right) \frac{\ln (2p)}{n}\right] \leq \exp (-y),
\label{probability bound} \\
& \mathbb{P}\left[ \Vert \widehat{\theta }-\theta _{\ast }\Vert _{0}\geq 
\frac{4-2\eta }{1-\eta }s+32\lambda ^{-1}C^{2}(s+k_{0})\left( \frac{1+\eta
+M\eta y}{1-\eta }\right) \frac{\ln (2p)}{n}\right] \leq \exp (-y),
\label{probability bound hamming}
\end{align}%
where%
\begin{equation}
C\equiv 8LB\kappa _{1}^{-1}\kappa _{0}^{-1},  \label{C0}
\end{equation}%
provided that 
\begin{equation}
\ln (2p)\geq \left( \frac{\kappa _{1}^{2}\kappa _{0}^{2}}{64L}\vee 1\right) .
\label{condition on p}
\end{equation}
\end{theorem}

Results (\ref{probability bound}) and (\ref{probability bound hamming}) of
Theorem \ref{thm-main-1} are non-asymptotic and establish exponential
inequalities on the tail probabilities of the excess quantile predictive
risk $S(\widehat{\theta })-S\left( \theta _{\ast }\right) $ as well as the $%
\ell _{0}$-distance between the $\ell _{0}$-PQR estimator and the true
parameter value. Applying inequality (\ref{probability bound}), we can
obtain non-asymptotic upper bounds on a triplet of population quantities:
(i) the mean excess predictive risk $\mathbb{E}[S(\widehat{\theta }%
)-S(\theta _{\ast })]$; (ii) the mean-square regression function estimation
error $\mathbb{E}[R(\widehat{\theta })]$; (iii) the mean-square parameter
estimation error $\mathbb{E}[\Vert \widehat{\theta }-\theta _{\ast }\Vert
_{2}^{2}]$. The results concerning these bounds are given in the next
theorem.

\begin{theorem}
\label{mean excess risk and mean squared estimation error}Let Assumptions %
\ref{iden-cond}--\ref{design-cond-2} hold. Suppose that $s\leq k_{0}$. Given
condition (\ref{condition on p}) of Theorem \ref{thm-main-1}, there is a
universal constant $K$, which depends only on the constants $L$ and $B$,
such that the following bounds hold:%
\begin{eqnarray}
\mathbb{E}\left[ S(\widehat{\theta })-S\left( \theta _{\ast }\right) \right]
&\leq &4\lambda s+\frac{K(s+k_{0})}{\kappa _{1}^{2}\kappa _{0}^{2}}\frac{\ln
(2p)}{n},  \label{mean excess risk} \\
\mathbb{E}\left[ R(\widehat{\theta })\right] &\leq &\kappa _{0}^{-2}\left(
4\lambda s+\frac{K(s+k_{0})}{\kappa _{1}^{2}\kappa _{0}^{2}}\frac{\ln (2p)}{n%
}\right) ,  \label{mean squared prediction error} \\
\mathbb{E}\left[ \left\Vert \widehat{\theta }-\theta _{\ast }\right\Vert
_{2}^{2}\right] &\leq &\kappa _{1}^{-2}\kappa _{0}^{-2}\left( 4\lambda s+%
\frac{K(s+k_{0})}{\kappa _{1}^{2}\kappa _{0}^{2}}\frac{\ln (2p)}{n}\right) .
\label{mean squared estimation error}
\end{eqnarray}
\end{theorem}

If $k_0/s$ is bounded by a fixed constant, we can deduce from Theorem \ref%
{mean excess risk and mean squared estimation error} that 
\begin{align*}
\mathbb{E}\left[ S(\widehat{\theta})-S\left( \theta _{\ast }\right) \right]
& =O\left[ \left( \lambda +n^{-1}\ln p\right) s\right] , \\
\mathbb{E}\left[ R(\widehat{\theta})\right] & =O\left[ \left( \lambda
+n^{-1}\ln p\right) s\right] , \\
\mathbb{E}\left[ \left\Vert \widehat{\theta }-\theta _{\ast }\right\Vert
_{2}^{2}\right] & =O\left[ \left( \lambda +n^{-1}\ln p\right) s\right] ,
\end{align*}%
which suggests that the optimal $\lambda $ be of the following form: 
\begin{equation}  \label{optimal-lambda}
\lambda =C_{\lambda }\frac{\ln p}{n},
\end{equation}%
where $C_{\lambda }$ is a positive constant that needs to be chosen by a
researcher. Under \eqref{optimal-lambda} and the side condition that $k_0/s
\leq C_k$ for some fixed constant $C_k$, we have that 
\begin{equation*}
\mathbb{E}\left[ S(\widehat{\theta})-S\left( \theta _{\ast }\right) \right]
=O\left( \frac{s\ln p}{n}\right) ,\mathbb{E}\left[ R(\widehat{\theta})\right]
=O\left( \frac{s\ln p}{n}\right) \text{ and }\mathbb{E}\left[ \left\Vert 
\widehat{\theta }-\theta _{\ast }\right\Vert _{2}^2 \right] =O\left( \frac{%
s\ln p}{n}\;\right).
\end{equation*}

We now specialize Theorem \ref{mean excess risk and mean squared estimation
error} to quantile regression. The following corollary provides the main
results for $\ell_0$-PQR.

\begin{corollary}
\label{thm-cor-1} Assume that (i) \eqref{boundedness} holds and $k_{0}\in
\lbrack s,C_{k}s]$ for a fixed constant $C_{k}\geq 1$, (ii) $\lambda
=C_{\lambda }{\ln p}/{n}$ for a fixed constant $C_{\lambda }>0$, (iii) $\mathbb{E}|Y|<\infty$, (iv) $f_{U|X}\left(
z|x\right) $ is bounded below by $c_{u}>0$ for all $z$ in an open interval
containing $\left[ -B^{2}\left( k_{0}+s\right) ,B^{2}\left( k_{0}+s\right) %
\right] $ and for all $x$ in the support of $X$, (v) for any subset $%
J\subset \{1,...,p\}$ such that $\left\vert J\right\vert \leq \left(
k_{0}+s\right) $, the smallest eigenvalue of $\mathbb{E}\left(
X_{J}X_{J}^{\top }\right) $ is bounded below by a positive constant $\omega $%
. Then, there is a universal constant $\bar{K}$, which depends only on the
constants $B$, such that 
\begin{eqnarray}
\mathbb{E}\left[ S(\widehat{\theta })-S\left( \theta _{\ast }\right) \right]
&\leq &4C_{\lambda }\frac{s\ln p}{n}+\frac{\bar{K}\left( C_{k}+1\right) s}{%
c_{u}\omega }\frac{\ln (2p)}{n},  \label{mean excess risk qr} \\
\mathbb{E}\left[ R(\widehat{\theta })\right] &\leq &\frac{8C_{\lambda }}{%
c_{u}}\frac{s\ln p}{n}+\frac{2\bar{K}\left( C_{k}+1\right) s}{%
c_{u}^{2}\omega }\frac{\ln (2p)}{n},  \label{mean squared prediction qr} \\
\mathbb{E}\left[ \left\Vert \widehat{\theta }-\theta _{\ast }\right\Vert
_{2}^{2}\right] &\leq &\frac{8C_{\lambda }}{c_{u}\omega }\frac{s\ln p}{n}+%
\frac{2\bar{K}\left( C_{k}+1\right) s}{c_{u}^{2}\omega ^{2}}\frac{\ln (2p)}{n%
},  \label{mean squared estimation error qr}
\end{eqnarray}%
provided that 
\begin{equation}
\ln (2p)\geq \left( \frac{c_{u}\omega }{128}\vee 1\right) .
\label{condition on p spec}
\end{equation}
\end{corollary}

Corollary \ref{thm-cor-1} provides non-asymptotic bounds on the mean-square
regression function and parameter estimation errors as well as the excess
quantile prediction risk. The resulting rates of convergence are of order $%
s\ln p/n$, which is the same as those of $\ell _{1}$-PQR and non-convex penalized quantile regression (see, e.g., \citet{belloni2011} and \citet{Wang13} for earlier
results and \citet{wang2019wp} and \citet{wang2022} for the latest
results). These are nearly minimax optimal rates of convergence because it
is shown in \citet[Theorem
4.1(i)]{wang2019wp} that the minimax lower bound for $\mathbb{E}[\Vert 
\widehat{\theta }-\theta _{\ast }\Vert _{2}^{2}]$ is of order $s\ln (p/s)/n$%
. The optimal tuning parameter $\lambda $ in $\ell _{1}$-PQR is of order $%
\sqrt{\ln p/n}$, whereas it is of order ${\ln p/n}$ in $\ell _{0}$-PQR.

\begin{remark}
Instead of assuming condition (iv) in Corollary \ref{thm-cor-1}, one may
assume the regularity conditions that are similar to those imposed in \cite%
{belloni2011}: that is, $f_{U|X}\left( 0|x\right) $ is bounded below by a
positive constant for all $x$ in the support of $X$, $\partial f_{U|X}\left(
z|x\right) /\partial z$ exists and is bounded in absolute value by a
constant uniformly in $(z,x)$, and 
\begin{equation*}
\inf_{\theta \in \mathbb{B}(k_{0}):\theta \neq \theta _{\ast }}\frac{\left\{ 
\mathbb{E}\left[ \left\vert X^{\top }(\theta -\theta _{\ast })\right\vert
^{2}\right] \right\} ^{3/2}}{\mathbb{E}\left[ \left\vert X^{\top }(\theta
-\theta _{\ast })\right\vert ^{3}\right] }>0.
\end{equation*}%
The last condition is called the restricted nonlinearity condition %
\citep{belloni2011}. In a recent working paper, \citet{wang2019wp}
established theoretical results for $\ell _{1}$-PQR without relying on the
restricted nonlinearity condition. In fact, \citet{wang2019wp} only assumed
a uniform lower bound for $f_{U|X}\left( \cdot |x\right) $ in a neighborhood
of zero, which is weaker than condition (iv) in Corollary \ref{thm-cor-1}.
It is an open question whether we can verify Assumption \ref{key-cond} under
a weaker condition imposed in \citet{wang2019wp}.
\end{remark}

Using the method for proving Corollary \ref{thm-cor-1}, we can obtain the
following result for $\ell_0$-CQR.

\begin{corollary}
\label{thm-cor-2} Assume that (i) \eqref{boundedness} holds, (ii) $s \leq q$%
, (iii) $\mathbb{E}|Y|<\infty $, (iv) $f_{U|X}\left( z|x\right) $ is bounded
below by $c_{u}>0$ for all $z$ in an open interval containing $\left[
-B^{2}\left( q+s\right) ,B^{2}\left( q+s\right) \right] $ and for all $x$
in the support of $X$, (v) for any subset $J\subset \{1,...,p\}$ such that $%
\left\vert J\right\vert \leq \left( q+s\right) $, the smallest eigenvalue of 
$\mathbb{E}\left( X_{J}X_{J}^{\top }\right) $ is bounded below by a positive
constant $\omega $. Then, there is a universal constant $\tilde{K}$, which
depends only on the constant $B$, such that 
\begin{eqnarray}
\mathbb{E}\left[ S(\widetilde{\theta})-S\left( \theta _{\ast }\right) \right]
&\leq &\frac{\tilde{K}(s+q)}{c_{u}\omega }\frac{\ln (2p)}{n},
\label{mean excess risk qr tilde} \\
\mathbb{E}\left[ R(\widetilde{\theta})\right] &\leq &\frac{2\tilde{K}(s+q)}{%
c_{u}^{2}\omega }\frac{\ln (2p)}{n},
\label{mean squared prediction qr tilde} \\
\mathbb{E}\left[ \left\Vert \widetilde{\theta }-\theta _{\ast }\right\Vert
_{2}^{2}\right] &\leq & \frac{2\tilde{K}(s+q)}{c_{u}^{2}\omega ^{2}}\frac{%
\ln (2p)}{n},  \label{mean squared estimation error qr tilde}
\end{eqnarray}%
provided that \eqref{condition on p spec} holds.
\end{corollary}

Corollary \ref{thm-cor-2} shows that the $\ell_0$-CQR estimator is also
nearly minimax optimal, provided that the imposed sparsity $q$ is at least
as large as the true sparsity $s$ and that $q/s$ is bounded by a fixed
constant. Therefore, our theory predicts that $\ell_0$-PQR and $\ell_0$-CQR
would perform similarly in applications.

\subsection{Hamming Loss}

Applying (\ref{probability bound hamming}) of Theorem \ref{thm-main-1}, we
now derive a theoretical result regarding the $\ell _{0}$-PQR in terms of
expected Hamming loss. Specifically, the following theorem presents an upper
bound on the expectation of the Hamming distance between $\widehat{\theta }$
and $\theta _{\ast }$.

\begin{theorem}
\label{thm:sparsity} Let Assumptions \ref{iden-cond}--\ref{design-cond-2}
hold. Furthermore, \eqref{condition on p} holds, $k_{0}\in \lbrack s,C_{k}s]$
for a fixed constant $C_{k}\geq 1$, and $\lambda =C_{\lambda }{\ln (p)}/{n}$%
. For any given $\nu >0$, there exists a sufficiently large constant $%
C_{\lambda }$, which does not depend on $(s,n,p)$, such that%
\begin{equation*}
\mathbb{E}\left[ \frac{\Vert \widehat{\theta }-\theta _{\ast }\Vert _{0}}{s}%
\right] \leq (4+\nu ).
\end{equation*}
\end{theorem}

Note that $D_{H}(\widehat{\theta },\theta _{\ast })\equiv s^{-1}\Vert 
\widehat{\theta }-\theta _{\ast }\Vert _{0}$ is the Hamming
distance---normalized by diving it by $s$---between $\widehat{\theta }$ and $%
\theta _{\ast }$, that is, $s^{-1}$ times the number of elements of the $%
\ell _{0}$-PQR estimator that are different from the corresponding elements
of the true parameter vector. Theorem \ref{thm:sparsity} shows that $\mathbb{%
E}[D_{H}(\widehat{\theta },\theta _{\ast })]$ can be bounded by a constant
that is slightly larger than 4, provided that the tuning parameter $\lambda $
is suitably chosen. Note that 
\begin{equation*}
\mathbb{P}\left( \widehat{\theta }\neq \theta _{\ast }\right) =\mathbb{P}%
\left( \Vert \widehat{\theta }-\theta _{\ast }\Vert _{0}\geq 1\right) \leq 
\mathbb{E}\left[ \Vert \widehat{\theta }-\theta _{\ast }\Vert _{0}\right] .
\end{equation*}%
Thus, we do not expect that $\mathbb{E}[\Vert \widehat{\theta }-\theta
_{\ast }\Vert _{0}]$ can be small since it is impossible to make $\mathbb{P}(%
\widehat{\theta }\neq \theta _{\ast })$ small. Instead, what we obtain in
Theorem \ref{thm:sparsity} is that $\mathbb{E}[\Vert \widehat{\theta }%
-\theta _{\ast }\Vert _{0}]$ is bounded by $(4+\nu )s$, independent of $p$.

In view of Theorem~\ref{mean excess risk and mean squared estimation error},
Theorem~\ref{thm:sparsity} suggests that the estimated sparsity and the
selected set of covariates of $\ell _{0}$-PQR cannot be too distinct from $s$
and the true set of nonzero elements of $\theta _{\ast }$. By Theorem~\ref%
{thm:sparsity} , the resulting sparsity of $\ell _{0}$-PQR is likely to be
substantially smaller than $k_{0}$ with a suitable choice of $\lambda $ and $%
k_{0}$; therefore, we expect that the constraint $\theta \in \mathbb{B}%
(k_{0})$ in \eqref{penalized QR} will not be binding in practice. Moreover,
since the choice of $C_{\lambda }$ in Theorem~\ref{thm:sparsity} is
independent of $(s,n,p)$, the minimax optimal rates are still intact.

\begin{remark}
Using a simple Gaussian mean model, \citet{butucea2018} considered variable
selection under expected Hamming loss. They derived
sufficient and necessary conditions under which the following term converges
to zero (using our notation): 
\begin{equation}
\mathbb{E}\left[ \frac{1}{s}\sum_{j=1}^{p}\left\vert 1(\widehat{\theta }%
_{j}\neq 0)-1(\theta _{\ast ,j}\neq 0)\right\vert \right] ,
\label{recovery-term}
\end{equation}%
where $1(\cdot )$ is the indicator function and $\widehat{\theta }_{j}$ and $%
\theta _{\ast ,j}$, respectively, are the $j$-th elements of $\widehat{%
\theta }$ and $\theta _{\ast }$. Their conditions involve the size of the
smallest non-zero elements of a signal vector. It is an interesting future
research topic to investigate the behavior of \eqref{recovery-term} in $\ell
_{0}$-PQR.
\end{remark}

\section{Implementation of \texorpdfstring{$\ell _{0}$}{L0}-PQR}

\label{Sec:Computation}

\subsection{Computation through Mixed Integer Optimization}

\label{Mixed Integer Optimization}

The MIO approach is useful for solving variable selection problems with $%
\ell _{0}$-norm constraints or penalties 
\citep[see,
e.g.,][]{bertsimas2016,chen2018, chen2018arXiv}. Assume that the parameter
space $\Theta $ takes the form $\Theta =\prod\nolimits_{j=1}^{p}\left[ 
\underline{\theta }_{j},\overline{\theta }_{j}\right] ,$ where $\underline{%
\theta }_{j}$ and $\overline{\theta }_{j}$ are lower and upper parameter
bounds such that $-\infty <\underline{\theta }_{j}\leq \theta _{j}\leq 
\overline{\theta }_{j}<\infty $ for $j\in \{1,...,p\}$. We now present an
implementation of $\ell _{0}$-PQR, which builds on the method of mixed
integer linear programming (MILP). Specifically, the $\ell _{0}$-penalized
minimization problem (\ref{penalized QR}) can be equivalently reformulated
as the following MILP problem:%
\begin{align}
& \min_{\theta \in \mathbf{\Theta },\left( r_{i},s_{i}\right)
_{i=1}^{n},(d_{j})_{j=1}^{p}}\frac{1}{n}\sum\nolimits_{i=1}^{n}\left[ \tau
r_{i}+\left( 1-\tau \right) s_{i}\right] +\lambda
\sum\nolimits_{j=1}^{p}d_{j}  \label{MIO} \\
& \text{subject to}  \notag \\
& r_{i}-s_{i}=Y_{i}-X_{i}^{\top }\theta ,\text{ }i\in \{1,...,n\},
\label{sum of positive and negative parts} \\
& d_{j}\underline{\theta }_{j}\leq \theta _{j}\leq d_{j}\overline{\theta }%
_{j},\text{ }j\in \{1,...,p\},  \label{selection constraint} \\
& d_{j}\in \{0,1\},\text{ }j\in \{1,...,p\},  \label{indicator di} \\
& r_{i}\geq 0,\text{ }s_{i}\geq 0,\text{ }i\in \{1,...,n\},
\label{non-negativity constraints} \\
& \sum_{j=1}^{p}d_{j}\leq k_{0}.  \label{k0-constraints}
\end{align}

We now explain the equivalence between (\ref{penalized QR}) and (\ref{MIO}).
If we remove from the problem (\ref{MIO}) the second term of the
objective function as well as all the $(d_{1},...,d_{p})$ control variables
together with their constraints (\ref{selection constraint}) and (\ref%
{indicator di}), the resulting minimization problem reduces to the linear
programming reformulation of the standard linear quantile regression problem %
\citep[][Section 6.2]{Koenker2005}. In the presence of the penalty term and
the $(d_{1},...,d_{p})$ controls, the inequality and dichotomization
constraints (\ref{selection constraint}) and (\ref{indicator di}) ensure
that, whenever $d_{j}=0$, the value $\theta _{j}$ must also be zero and the
sum $\sum\nolimits_{j=1}^{p}d_{j}$ thus captures the number of non-zero
components of the vector $\theta $. The last constraint %
\eqref{k0-constraints} imposes that the estimated sparsity is at most $k_0$. 
As a result, both minimization problems (\ref{penalized QR}) and (\ref{MIO})
are equivalent. This equivalence enables us to employ modern MIO solvers to
solve $\ell _{0}$-PQR problems.

\subsection{Computation through First-Order Approximation}

\label{First-Order Approximation}

The MIO formulation (\ref{MIO}) is concerned with optimization over
integers, which could be computationally challenging for large scale
problems. \citet[][Section 3]{bertsimas2016} have developed discrete
first-order algorithms enabling fast computation of near optimal solutions
to $\ell _{0}$-constrained least squares and least absolute deviation
estimation problems. \citet{Huang:2018} have also proposed fast and scalable
algorithms for computing approximate solutions to $\ell _{0}$-penalized
least squares estimation problems. These algorithms build on the necessary
conditions for optimality in the $\ell _{0}$-constrained or penalized
optimization problems. Motivated from these papers, in this subsection, we
present a first-order approximation algorithm that can be used as either a
standalone solution algorithm or a warm-start strategy for enhancing the
computational performance of our MIO approach to the $\ell _{0} $-PQR
problem.

For $\tau \in (0,1)$, the quantile regression objective function (\ref{Sn})
can be equivalently expressed as%
\begin{equation}
S_{n}(\theta )=n^{-1}\max_{\tau -1\leq w_{i}\leq \tau
}\sum_{i=1}^{n}w_{i}(Y_{i}-X_{i}^{\top }\theta ).
\label{Sn equivalent formulation}
\end{equation}%
The function $S_{n}(\theta )$ is nonsmooth. Following \citet{nesterov2005},
we can construct a smooth approximation of $S_{n}(\theta )$ by 
\begin{equation}
S_{n}(\theta ;\delta )\equiv n^{-1}\max_{\tau -1\leq w_{i}\leq \tau }\left[
\sum_{i=1}^{n}w_{i}(Y_{i}-X_{i}^{\top }\theta )-\frac{\delta }{2}\left\Vert
w\right\Vert _{2}^{2}\right]  \label{Sn(theta;delta)}
\end{equation}%
where $w$ denote the vector of controls $\left( w_{1},...,w_{n}\right) $ in
the maximization problem (\ref{Sn(theta;delta)}). Note that %
\citet{nesterov2005}'s smoothing method is different from a
convolution-based smoothing method for quantile regression by %
\citet{FGH:2021} and \citet{HE2023}.

Assume that the parameter space $\Theta $ is of an equilateral cube form $%
\Theta =\left[ -B,B\right] ^{p}$ for some $B>0$. Let $t$ be any given vector
in $\mathbb{R}^{p}$. Let $\widehat{\beta }$ be a solution to the following $%
\ell _{0}$-penalized minimization problem:%
\begin{equation}
\min\nolimits_{\beta \in \mathbb{B}(k_{0})}\left\Vert \beta -t\right\Vert
_{2}^{2}+\lambda \Vert \beta \Vert _{0},  \label{thresholding rule}
\end{equation}%
where $\lambda $ is a non-negative penalty tuning parameter. It is
straightforward to see that the solution $\widehat{\beta }$ can be computed
as follows. Let $\widetilde{\beta }$ be a $p$ dimensional vector given by%
\begin{equation*}
\widetilde{\beta }_{j}=\left\{ 
\begin{array}{l}
B \times 1\left\{ B^{2}-2t_{j}B+\lambda <0\right\} \text{ if }t_{j}>B \\ 
t_{j} \times 1\{\left\vert t_{j}\right\vert >\sqrt{\lambda }\}\text{ if }%
-B\leq t_{j}\leq B \\ 
-B \times 1 \left\{ B^{2}+2t_{j}B+\lambda <0\right\} \text{ if }t_{j}<-B%
\end{array}%
\right. ,
\end{equation*}
for $j\in \{1,...,p\}$. Then the solution $\widehat{\beta }=$ $\widetilde{%
\beta }$ if $\Vert \widetilde{\beta }\Vert _{0}\leq k_{0}$. Otherwise,
letting $S(t)$ denote the set of $k_{0}$ indices that keep track of the
largest $k_{0}$ components of $t$ in absolute value, we set $\widehat{\beta }%
_{j}=\widetilde{\beta }_{j}$ for $j\in S(t)$ and $\widehat{\beta }_{j}=0$
for $j\notin S(t)$. Therefore, the problem (\ref{thresholding rule}) admits
a simple closed-form solution. We will exploit this fact and develop a
first-order approximation algorithm.

Define%
\begin{equation}
Q_{n}(\theta ;\delta )\equiv S_{n}(\theta ;\delta )+\lambda \Vert \theta
\Vert _{0}.  \label{Qn(theta;delta)}
\end{equation}%
For any vector $t\in \mathbb{R}^{p}$, suppose we can construct a quadratic
envelope of $S_{n}(\theta ;\delta )$ with respect to the vector $t$ in the
sense that%
\begin{equation}
S_{n}(\theta ;\delta )\leq \widetilde{S}_{n}(\theta ;t,\delta ,l)\equiv
S_{n}(t;\delta )+\bigtriangledown _{\theta }S_{n}(t;\delta )^{\top }\left(
\theta -t\right) +\frac{l}{2}\left\Vert \theta -t\right\Vert _{2}^{2}
\label{quadratic envelope}
\end{equation}%
for some non-negative real scalar $l$, which does not depend on the
parameter vector $\theta $. Note that (\ref{quadratic envelope}) holds
whenever the gradient function $\bigtriangledown _{\theta }S_{n}(\cdot
;\delta )$ is Lipschitz continuous such that 
\begin{equation}
\Vert \bigtriangledown _{\theta }S_{n}(t;\delta )-\bigtriangledown _{\theta
}S_{n}(t^{\prime };\delta )\Vert _{2}\leq h\Vert t-t^{\prime }\Vert _{2}
\label{Lipschitz}
\end{equation}%
for some Lipschitz constant $h$, which does not depend on $t$ and $t^{\prime
}$. By the envelope theorem, 
\begin{equation*}
\bigtriangledown _{\theta }S_{n}(t;\delta )=-\frac{1}{n}\sum%
\nolimits_{i=1}^{n}X_{i}\widehat{w}_{i,\delta },
\end{equation*}%
where $(\widehat{w}_{1,\delta },...,\widehat{w}_{n,\delta })$ is the
solution to the minimization problem (\ref{Sn(theta;delta)}). Using %
\citet[][Theorem 1]{nesterov2005}, we can deduce that (\ref{Lipschitz})
holds with 
\begin{equation}
h=\frac{1}{n\delta }\;\text{trace}\left(
\sum\nolimits_{i=1}^{n}X_{i}X_{i}^{\prime }\right)
\label{Lipschitz constant}
\end{equation}%
and hence (\ref{quadratic envelope}) holds for every $l\geq h$.

Define 
\begin{equation*}
\widetilde{Q}_{n}(\theta ;t,\delta ,l)\equiv \widetilde{S}_{n}(\theta
;t,\delta ,l)+\lambda \Vert \theta \Vert _{0}.
\end{equation*}%
Note that $\widetilde{Q}_{n}(\theta ;t,\delta ,l)$ is an upper envelope of $%
Q_{n}(\theta ;\delta )$ around the vector $t$ with the property that $%
\widetilde{Q}_{n}(t;t,\delta ,l)=Q_{n}(t;\delta )$.

For $t\in \mathbb{R}^{p}$, define the mapping%
\begin{equation}
H_{\delta ,l}(t)\equiv \arg \min_{\theta \in \mathbb{B}(k_{0})}\left\{
\left\Vert \theta -\left( t-\frac{1}{l}\bigtriangledown _{\theta
}S_{n}(t;\delta )\right) \right\Vert _{2}^{2}+\lambda \Vert \theta \Vert
_{0}\right\} .  \label{H(t)}
\end{equation}%
Arranging the terms, we can easily deduce%
\begin{equation}
H_{\delta ,l}(t)=\arg \min_{\theta \in \mathbb{B}(k_{0})}\widetilde{Q}%
_{n}(\theta ;t,\delta ,l).  \label{H(t) and Q(b)}
\end{equation}%
We say that a point $t\in \mathbb{R}^{p}$ is a stationary point of the
mapping $H_{\delta ,l}$ if $t\in H_{\delta ,l}(t)$. For each given value of $%
\delta $, let $\widehat{\theta }_{\delta }$ denote a solution to the problem
of minimizing $Q_{n}(\theta ;\delta )$ over $\theta \in \mathbb{B}(k_{0})$.
We propose to approximate $\widehat{\theta }_{\delta }$ by solving for the
stationary point of the mapping $H_{\delta ,l}$. This can be justified by
the following proposition, which is a straightforward extension of Theorem
3.1 of \citet{bertsimas2016}.

\begin{proposition}[\citet{bertsimas2016}]
\label{stationary point}The following statements hold:

(a) If $\widehat{\theta }_{\delta }\in \arg \min_{\theta \in \mathbb{B}%
(k_{0})}Q_{n}(\theta ;\delta )$, then $\widehat{\theta }_{\delta }\in
H_{\delta ,l}(\widehat{\theta }_{\delta })$.

(b) Let $l>h$ and $t_{m}$ be a sequence such that $t_{m+1}\in H_{\delta
,l}(t_{m})$. Then, for some limits $t^{\ast }$ and $Q^{\ast }$, we have that 
$t_{m}\longrightarrow t^{\ast }$, $Q_{n}(t_{m};\delta )\downarrow Q^{\ast }$
as $m\longrightarrow \infty $. Moreover, 
\begin{equation}
\min_{m=1,...,N}\Vert t_{m+1}-t_{m}\Vert _{2}^{2}\leq \frac{2\left(
Q_{n}(t_{1};\delta )-Q^{\ast }\right) }{N\left( l-h\right) }.
\label{convergence}
\end{equation}
\end{proposition}

Proposition \ref{stationary point} implies that any solution to the
minimization of $Q_{n}(\theta ;\delta )$ over $\theta \in \mathbb{B}(k_{0})$
is also a stationary point of the mapping $H_{\delta ,l}$. Moreover we can
solve for a stationarity point by iterating until convergence. Result (\ref%
{convergence}) indicates that the convergence rate is $O(N^{-1})$, where $N$
is the number of performed iterations. Note that we can use (\ref%
{thresholding rule}) to obtain a closed-form solution to the $\ell _{0}$%
-penalized minimization problem (\ref{H(t)}) for every $t\in \mathbb{R}^{p}$
and therefore solving for a stationary point of $H_{\delta ,l}$ would incur
relatively little computational cost.

We now turn to the $\ell _{0}$-PQR problem (\ref{penalized QR}). The next
proposition builds on the results of \citet{nesterov2005} concerning the
uniform approximation bound of the smooth function to the non-smooth
quantile loss function. Let $c_{\tau }\equiv \tau ^{2}\vee \left( 1-\tau
\right) ^{2}$.

\begin{proposition}
\label{Approximation}For $\delta \geq 0$, if $\widehat{\theta }_{\delta }\in
\arg \min_{\theta \in \mathbb{B}(k_{0})}Q_{n}(\theta ;\delta )$, then 
\begin{equation}
S_{n}(\widehat{\theta }_{\delta })+\lambda \Vert \widehat{\theta }_{\delta
}\Vert _{0}\leq \min\nolimits_{\theta \in \mathbb{B}(k_{0})}\left\{
S_{n}(\theta )+\lambda \Vert \theta \Vert _{0}\right\} +\frac{\delta c_{\tau
}}{2}.  \label{approximation bound}
\end{equation}
\end{proposition}

Given a tolerance level $\epsilon $, Proposition \ref{Approximation} implies
that, for any given $\delta \leq 2\epsilon c_{\tau }^{-1}$, if we solve for
the minimization of $Q_{n}(\theta ;\delta )$ over $\theta \in \mathbb{B}%
(k_{0})$, the resulting solution $\widehat{\theta }_{\delta }$ is an $%
\epsilon $-level approximate $\ell _{0}$-PQR estimator in the sense that%
\begin{equation*}
S_{n}(\widehat{\theta }_{\delta })+\lambda \Vert \widehat{\theta }_{\delta
}\Vert _{0}\leq \min\nolimits_{\theta \in \mathbb{B}(k_{0})}\left\{
S_{n}(\theta )+\lambda \Vert \theta \Vert _{0}\right\} +\epsilon \text{.}
\end{equation*}%
This thus yields the following algorithm for computing a near optimal
solution to the $\ell _{0}$-PQR problem (\ref{penalized QR}).

\begin{algorithm}
\label{First Order Approximation Algorithm}Given an initial guess $\widehat{%
\theta }_{1}$, set $\delta =2\epsilon c_{\tau }^{-1}$ and perform the
following iterative procedure starting with $k=1$:

\begin{enumerate}
\item[Step 1.] For $k\geq 1$, compute $\widehat{\theta }_{k+1}\in H_{\delta
,l}(\widehat{\theta }_{k})$.

\item[Step 2.] Repeat Step 1 until the objective function $Q_{n}(\cdot
;\delta )$ converges.
\end{enumerate}
\end{algorithm}

We end this section with a remark that unlike MIO, the first-order
approximation method only delivers a feasible solution to the minimization
of $Q_{n}(\theta ;\delta )$ over $\theta \in \mathbb{B}(k_{0})$ and this
solution does not necessarily coincide with a global optimal solution to the 
$\ell _{0}$-PQR problem (\ref{penalized QR}).

\section{Simulation Study\label{Sec:Simulation}}

In this section, we perform Monte Carlo simulation experiments to evaluate
the performance of our $\ell _{0}$-based quantile regression approaches. We
consider the following data generating setup. Let $Z=(Z_{1},...,Z_{p-1})$ be
a $p-1$ dimensional multivariate normal random vector with mean zero and
covariance matrix $\Sigma $ with its element $\Sigma _{i,j}=\left(
0.5\right) ^{\left\vert i-j\right\vert }$. Let $X=(X_{1},...,X_{p})$ be a $p$
dimensional covariate vector with its components $X_{1}=1$ and $%
X_{j}=Z_{j-1}1\{\left\vert Z_{j-1}\right\vert \leq 6\}$ for $j\in \{2,..,p\}$%
. The outcome $Y$ is generated according to the model:%
\begin{equation*}
Y=X^{\top }\theta _{\ast }+X_{2}\varepsilon ,
\end{equation*}%
where $\varepsilon $ is a random disturbance which is independent of $X$ and
follows the univariate normal distribution with mean zero and standard
deviation $0.25$. We considered two configurations of the true parameter
vector $\theta _{\ast }$. For configuration (i), we set the sparsity $s=5$
and the true parameter value $\theta _{\ast ,j}=1$ for $s$ equispaced
values. For configuration (ii), we employed a more challenging case where $%
s=20$ and all the $s$ nonzero components of $\theta _{\ast }$ occurred at
equispaced indices between 1 and $p$ with its first 5 components equal to 1
and remaining $\left( s-5\right) $ nonzero components being set to be $%
\left( 2^{-1},2^{-2},...,2^{-(s-5)}\right) $, which decreased exponentially
to zero.

We compared the finite-sample performance among the $\ell _{0}$-PQR and $%
\ell _{0}$-CQR of the present paper, the $\ell _{1}$-PQR of %
\citet{belloni2011}, the adaptive Lasso penalized quantile regression of %
\citet{Fan14adaptive} and the nonconvex penalized quantile regression of %
\citet{Wang12}. In each simulation repetition, we generated a training
sample of $n=100$ observations for estimating the parameter vector $\theta $
and another independent validation sample of $100$ observations for
calibrating the tuning parameters of these estimation approaches. Moreover,
we also generated a test sample of $5000$ observations for evaluating the
out-of-sample predictive performance.

We focused our simulation study on median regression ($\tau =0.5$). To
implement the $\ell _{1}$-PQR approach, we used the $\ell _{1}$-penalized
quantile regression estimator of \citet{belloni2011} with the penalty level
given by 
\begin{equation}
\lambda _{BC}\equiv c_{BC}\Lambda \left( 1-\alpha |X\right) ,
\label{L1 PQR tuning parameter}
\end{equation}%
where $\Lambda \left( 1-\alpha |X\right) $ is the $\left( 1-\alpha \right) $
level quantile of the random variate $\Lambda $, which is defined in %
\citet[][equation (2.6)]{belloni2011}, conditional on the covariate vector $%
X $. Following \citet{belloni2011}, we set $\alpha =0.1$. Moreover, we
calibrated the optimal tuning value $c_{BC}$ from a set of candidate values $%
\mathcal{S}$ using the aforementioned validation sample in the setup with $%
p\geq 100$. For the low dimensional setup with $p<100$, we performed this
calibration over an expanded set $\mathcal{S\cup \{}0\}$, thereby allowing
for an estimating model that did not penalize any parameter. For simulations
under parameter configuration (i), we set $\mathcal{S}=\{0.1,0.2,...,1.9,2\}$%
. Under parameter configuration (ii), which is a more difficult case for
estimation, we further enlarged the tuning value search space by taking $%
\mathcal{S}$ to be $\{0.01,0.02,...1.99,2\}$.

To implement the $\ell _{0}$-CQR method, we solved over the training sample
the $\ell _{0}$-constrained estimation problem (\ref{L0-CQR-est}) for
sparsity level $q$ ranging from $1$ up to $p\wedge 25$. To solve (\ref%
{L0-CQR-est}) with $\tau =0.5$ for a given value of $q$, following %
\citet[][Section 6]{bertsimas2016}, we used the MIO-based, $\ell _{0}$%
-constrained LAD approach with a warm-start strategy by supplying the MIO
solver an initial guess computed via the discrete first-order approximation
algorithms. We then calibrated the optimal sparsity level among this set of $%
q$ values using the calibration sample. The resulting $\ell _{0}$-CQR
estimator was then constructed based on the model associated with the
calibrated optimal sparsity level.

For the $\ell _{0}$-PQR method, noting that the scale of the quantile
regression objective function $S_{n}(\theta )$ varies whenever that of $Y$
changes, to relate the penalty term to the scale of $Y$ and to the derived
rate \eqref{optimal-lambda}, we adopted the following simple rule: 
\begin{equation}
\lambda =c\left( n^{-1}\sum\nolimits_{i=1}^{n}|Y_{i}|\right) \frac{\ln p}{n},
\label{L0 penalty tuning parameter}
\end{equation}%
which is proportional to the sample average of the absolute value of $Y$.

For a given value of $c$ in (\ref{L0 penalty tuning parameter}), we solved
the problem (\ref{penalized QR}) with $k_{0}=100$ using our MIO
computational approach of Section \ref{Sec:Computation}, where we
warm-started the MIO solver by supplying as an initial guess the approximate
solution obtained through the first-order method of Section \ref{First-Order
Approximation}. As in the $\ell _{1}$-PQR case, we calibrated the optimal
tuning scalar $c$ over the set $\mathcal{S}$ using the calibration sample in
the setup with $p\geq k_{0}$ and over the expanded set $\mathcal{S\cup \{}%
0\} $ in the setup with $p<k_{0}$.

We provided further details here on the implementation of our first-order
approximation procedure in Algorithm \ref{First Order Approximation
Algorithm}. We set the tolerance level $\epsilon $ to be $2\cdot 10^{-4}$
and parameter $l$ of the quadratic envelope in (\ref{quadratic envelope}) to
be $2h$, where $h$ is the Lipschitz constant given by (\ref{Lipschitz
constant}). Note that Algorithm \ref{First Order Approximation Algorithm}
also requires an initial guess. We therefore ran it for $T=50$ times, each
of which was performed with a different initial guess and used the output
that delivered the best penalized objective function value in (\ref%
{penalized QR}) as the resulting first-order approximate solution. We chose
these $T$ initial guesses sequentially where the first one was the $\ell
_{1} $-PQR solution of \citet{belloni2011} implemented with its tuning value 
$c_{BC}$ set to be identical to the given value $c$ in (\ref{L0 penalty
tuning parameter}) whereas, for $t\in \{2,...,T\}$, the $t$-th initial guess
was subsequently constructed as the solution to the standard quantile
regression of the outcome $Y$ on those covariates selected in the output of
Algorithm \ref{First Order Approximation Algorithm} which was initiated with
the $\left( t-1\right) $-th initial guess. We found this implementation
procedure worked very well in both our simulation study here and the
empirical application of Section \ref{Sec:Application}.

We specified the parameter space $\Theta $ to be $[-10,10]^{p}$ for the MIO
computation of both the $\ell _{0}$-PQR and $\ell _{0}$-CQR estimators.
Throughout this paper, we used the MATLAB implementation of the Gurobi
Optimizer (version 8.1.1) to solve all the MIO problems. Moreover, all
numerical computations were done on a desktop PC (Windows 7) equipped with
128 GB RAM and a CPU processor (Intel i9-7980XE) of 2.6 GHz. To reduce
computation cost in all MIO computations associated with the covariate
configuration of $p=500$, we set the MIO solver time limit to be 10 minutes
beyond which we forced the solver to stop early and used the best discovered
feasible solution to construct the resulting $\ell _{0}$-PQR and $\ell _{0}$%
-CQR estimators.

We also compared our $\ell _{0}$-based estimators with the adaptive Lasso
quantile regression approach of \citet{Fan14adaptive}. Specifically, the
latter approach seeks to minimize the following weighted $\ell _{1}$%
-penalized quantile regression objective function%
\begin{equation}
S_{n}(\theta )+\sum_{j=1}^{p}g_{\mu }\left( \left\vert \widehat{\theta }%
_{j}^{ini}\right\vert \right) \left\vert \theta _{j}\right\vert ,
\label{adaptive Lasso}
\end{equation}%
where $\widehat{\theta }^{ini}=(\widehat{\theta }_{1}^{ini},...,\widehat{%
\theta }_{p}^{ini})$ is an initial high dimensional quantile regression
estimator and $g_{\mu }$ is a penalty weight function. For implementation,
we set $\widehat{\theta }^{ini}$ to be the $\ell _{1}$-PQR estimator and
considered the following two choices for the penalty weight $g_{\mu }$. The
first choice, which is based on the derivative of the smoothly clipped
absolute deviation (SCAD) penalty function \citep{fan2001}, is given by 
\begin{equation}
g_{\mu }(t)=\mu 1\{t\leq \mu \}+\frac{\left( a\mu -t\right) \vee 0}{a-1}%
1\{t>\mu \}  \label{SCAD weight}
\end{equation}%
for some parameters $a>2$ and $\mu \geq 0$. For ease of reference, we use
AL-SCAD as shorthand for the adaptive Lasso quantile regression approach
implemented with the SCAD based penalty weight (\ref{SCAD weight}). For the
second choice, we specify $g_{\mu }$ to be the derivative of the minimax
concave penalty (MCP) function \citep{zhang2010}, which is given by%
\begin{equation}
g_{\mu }(t)=\left( \mu -\frac{t}{a}\right) 1\{t\leq a\mu \}
\label{MCP weight}
\end{equation}%
for some parameters $a>1$ and $\mu \geq 0$. We refer to AL-MCP as shorthand
for the estimation approach based on the minimization of (\ref{adaptive
Lasso}) with the MCP based penalty weight (\ref{MCP weight}).

Finally, we considered the approach of nonconvex penalized quantile
regression of \citet{Wang12} implemented with either the SCAD or MCP
penalty. Specifically, this approach is based on minimizing the penalized
objective function 
\begin{equation}
S_{n}(\theta )+\sum_{j=1}^{p}G_{\mu }\left( \left\vert \theta
_{j}\right\vert \right) ,  \label{fully nonconvex penalized QR}
\end{equation}%
where $G_{\mu }$ is either the SCAD or the MCP penalty function, both of
which are nonconvex. Following \citet{zou2008} and \citet{Wang12}, we
adopted the local linear approximation algorithm to solve this nonconvex
minimization problem. This algorithm proceeds as follows. Let $\widehat{%
\theta }^{(0)}=(\widehat{\theta }_{1}^{(0)},...,\widehat{\theta }_{p}^{(0)})$
be an initial estimator, which we take as the $\ell _{1}$-PQR estimator.
Given an estimator $\widehat{\theta }^{(m)}$ at the $m$th iteration stage,
we solve for $\widehat{\theta }^{(m+1)}$ by minimizing (\ref{adaptive Lasso}%
) with $\widehat{\theta }^{ini}$ being replaced by $\widehat{\theta }^{(m)}$
and $g_{\mu }$, which is the derivative of $G_{\mu }$, taking the form (\ref%
{SCAD weight}) if $G_{\mu }$ is the SCAD penalty or (\ref{MCP weight}) if $%
G_{\mu }$ is the MCP penalty function. We then iterate this process until
the vector of weight differences $\left[ g_{\mu }\left( \left\vert \widehat{%
\theta }_{j}^{(m+1)}\right\vert \right) -g_{\mu }\left( \left\vert \widehat{%
\theta }_{j}^{(m)}\right\vert \right) \right] $ converges in $\ell _{2}$%
-norm within a numerical tolerance of $10^{-4}$. We refer to QR-SCAD and
QR-MCP as shorthand for the nonconvex penalized quantile approaches
implemented respectively with the SCAD and MCP penalties. Throughout the
implementation of all the methods that use (\ref{SCAD weight}) and (\ref{MCP
weight}), we set $a=3.7$ and focused on the calibration of the tuning
parameter $\mu $, which was performed over the set $\mathcal{S}$ using the
calibration sample in the setup with $p\geq 100$ and over the expanded set $%
\mathcal{S\cup \{}0\}$ in the setup with $p<100$.

We reported performance results based on 100 simulation repetitions. We
considered the following performance measures. Abusing the notation a bit,
let $\widehat{\theta }$ denote the estimated parameters under a given
quantile regression approach. To assess the predictive performance, we
reported the relative risk, which is the ratio of the median predictive risk
evaluated at the estimate $\widehat{\theta }$ over that evaluated at the
true value $\theta _{\ast }$. We approximated the out-of-sample predictive
risk using the generated 5000-observation test sample. Let $in\_RR$ and $%
out\_RR$ respectively denote the average of in-sample and that of
out-of-sample relative risks over the simulation repetitions.

We also reported the estimation performance in terms of both the average
parameter estimation error defined as $\mathbb{E}[ \Vert \widehat{\theta }%
-\theta _{\ast }\Vert _{2}] $ and the average regression function estimation
error defined as $\mathbb{E}[|X^{\top }(\widehat{\theta }-\theta _{\ast
})|^{2}]$. Finally, we examined the variable selection performance. We say
that a covariate $X_{j}$ is effectively selected if and only if the
magnitude of $\widehat{\theta }_{j}$ is larger than a small tolerance level
(e.g., $10^{-5}$ as used in our numerical study) which is distinct from zero
in numerical computation. Let $Avg\_sparsity$ denote the average number of
effectively selected covariates. Let $Corr\_sel$ be the proportion of the
truly relevant covariates being effectively selected. Let $Orac\_sel$ be the
proportion of obtaining an oracle variable selection outcome where the set
of effectively selected covariates coincides exactly with that of the truly
relevant covariates. Finally, let $Num\_irrel$ denote the average number of
effectively selected covariates whose true regression coefficients are zero.

\subsection{Simulation Results under Parameter Configuration (i)}

For parameter configuration (i), we performed simulations with $p\in
\{10,500\}$ to assess the performance in both the low and high dimensional
settings. The results for these two settings are presented respectively in
Tables \ref{Table 1} and \ref{Table 2}. For $\ell _{0}$-PQR, we report
performance measures for both the implementation based on the first-order
(FO) approximation and that based on the MIO, which was warm-started by
using the FO solutions as initial guesses. We find that, regarding the
predictive performance, all the competing approaches performed comparably
well for both in-sample and out-of-sample relative risks under the low
dimensional covariate design. By contrast, for the high dimensional design, $%
\ell _{1}$-PQR was considerably dominated by all the other approaches in
terms of out-of-sample predictive performance.

\begin{center}
\begin{table}[tbh]
\caption{Simulation comparison for $p=10$ under parameter configuration (i)}
\label{Table 1}{\small 
\begin{tabular}{c|cccccccc}
\hline\hline
$p=10$ & \multicolumn{2}{|c}{$\ell _{0}$-PQR} & $\ell _{0}$-CQR & $\ell _{1}$%
-PQR & AL-SCAD & AL-MCP & QR-SCAD & QR-MCP \\ 
& MIO & FO &  &  &  &  &  &  \\ \hline
$Corr\_sel$ & \multicolumn{1}{|l}{1} & 1 & 1 & 1 & 1 & 1 & 1 & 1 \\ 
$Orac\_sel$ & \multicolumn{1}{|l}{0.78} & 0.77 & 0.6 & 0.05 & 0.8 & 0.79 & 
0.79 & 0.81 \\ 
$Num\_irrel$ & \multicolumn{1}{|l}{1.10} & 1.11 & 1.02 & 3.13 & 0.92 & 0.93
& 0.97 & 0.87 \\ 
$Avg\_sparsity$ & \multicolumn{1}{|l}{6.10} & 6.11 & 6.02 & 8.13 & 5.92 & 
5.93 & 5.97 & 5.87 \\ 
$\mathbb{E}\left[ \left\Vert \widehat{\theta }-\theta _{\ast }\right\Vert
_{2}\right] $ & \multicolumn{1}{|l}{0.035} & 0.036 & 0.038 & 0.047 & 0.034 & 
0.034 & 0.035 & 0.034 \\ 
$\mathbb{E}[|X^{\top }(\widehat{\theta }-\theta _{\ast })|^{2}]$ & 
\multicolumn{1}{|l}{0.001} & 0.001 & 0.001 & 0.002 & 0.001 & 0.001 & 0.001 & 
0.001 \\ 
$in\_RR$ & \multicolumn{1}{|l}{0.976} & 0.976 & 0.974 & 0.969 & 0.980 & 0.979
& 0.979 & 0.979 \\ 
$out\_RR$ & \multicolumn{1}{|l}{1.029} & 1.030 & 1.030 & 1.040 & 1.028 & 
1.027 & 1.029 & 1.027 \\ \hline
\end{tabular}%
}
\end{table}

\begin{table}[tbh]
\caption{Simulation comparison for $p=500$ under parameter configuration (i)}
\label{Table 2}{\small 
\begin{tabular}{c|cccccccc}
\hline\hline
$p=500$ & \multicolumn{2}{|c}{$\ell _{0}$-PQR} & $\ell _{0}$-CQR & $\ell
_{1} $-PQR & AL-SCAD & AL-MCP & QR-SCAD & QR-MCP \\ 
& MIO & FO &  &  &  &  &  &  \\ \hline
$Corr\_sel$ & 1 & 1 & 1 & 1 & 1 & 1 & 1 & 1 \\ 
$Orac\_sel$ & 1 & 1 & 0.97 & 0 & 0.84 & 0.89 & 0.84 & 0.87 \\ 
$Num\_irrel$ & 0 & 0 & 0.03 & 29.26 & 2.11 & 1.69 & 2.11 & 1.68 \\ 
$Avg\_sparsity$ & 5 & 5 & 5.03 & 34.26 & 7.11 & 6.69 & 7.11 & 6.68 \\ 
$\mathbb{E}\left[ \left\Vert \widehat{\theta }-\theta _{\ast }\right\Vert
_{2}\right] $ & \multicolumn{1}{|l}{0.028} & 0.028 & 0.028 & 0.135 & 0.029 & 
0.029 & 0.028 & 0.028 \\ 
$\mathbb{E}[|X^{\top }(\widehat{\theta }-\theta _{\ast })|^{2}]$ & 
\multicolumn{1}{|l}{0.001} & 0.001 & 0.001 & 0.019 & 0.001 & 0.001 & 0.001 & 
0.001 \\ 
$in\_RR$ & \multicolumn{1}{|l}{0.981} & 0.981 & 0.981 & 0.793 & 0.967 & 0.970
& 0.967 & 0.969 \\ 
$out\_RR$ & \multicolumn{1}{|l}{1.024} & 1.024 & 1.025 & 1.282 & 1.026 & 
1.026 & 1.025 & 1.024 \\ \hline
\end{tabular}%
}
\end{table}
\end{center}

Turning to the variable selection results, we note that all the eight
estimation approaches had perfect $Corr\_sel$ rates and hence were effective
for selecting the relevant covariates. However, superb $Corr\_sel$
performance might just be a consequence of overfitting, which may result in
excessive selection of irrelevant covariates and adversely impact on the
out-of-sample predictive performance. From the results on the variable
selection performance measures, we note that the number of irrelevant
variables selected under $\ell _{1}$-PQR was quite large relatively to those
under the other seven approaches in the high dimensional setup even though
all of the considered estimation approaches exhibited the effect of reducing
the covariate space dimension. This echoes with the finding in the setup of $%
p=500$ that $\ell _{1}$-PQR had far better in-sample fit in terms of $in\_RR$
yet worse out-of-sample fit in terms of $out\_RR$ relatively to all the
other quantile regression approaches. Besides, while we could observe
nonzero and high values of $Orac\_Sel$ for the $\ell _{0}$-PQR, $\ell _{0}$%
-CQR, adaptive Lasso based and other nonconvex penalized estimation
approaches in both the low and high dimensional setups, the $\ell _{1}$-PQR
approach could rarely induce oracle variable selection outcome in these
simulations. We also find, except for $\ell _{1}$-PQR, which performed
relatively poorly, all the other approaches performed comparably well in the
parameter and regression function estimation performances. Finally, for the $%
\ell _{0}$-PQR approach, the FO-based algorithm as a standalone solution
algorithm also performed very well. In fact, we find that in the high
dimensional setup, all the MIO-based $\ell _{0}$-PQR computations could not
converge within the 10-minute computational time limit and the discovered
solutions upon early stopping coincided with the FO-based solutions which
were used to warm-start the MIO solver. This indicates that the first-order
algorithm already located high-quality $\ell _{0}$-PQR solutions upon which
further improvements through the global optimization solver of MIO could not
be obtained within the given computational time constraint. These simulation
results thus shed lights on the usefulness of our first-order approximation
approach for solving $\ell _{0}$-PQR problems in the presence of
computational resource constraints.

\subsection{Simulation Results under Parameter Configuration (ii)}

In this section we report results of simulations conducted with $p=500$
under parameter configuration (ii) where $s=20$ and the true nonzero
regression coefficients can decay exponentially to zero. Under this
simulation design, it turns out that none of the estimation approaches could
yield nonzero $Corr\_sel$ and $Orac\_Sel$ values. This result is not
surprising as the true coefficients of some relevant covariates are of very
small magnitudes so that none of the estimation methods in this simulation
study could effectively select these variables. We now summarize in Table 3
the results of remaining performance measures. From this table, it is
evident that $\ell _{1}$-PQR tended to select far more irrelevant covariates
and thus performed quite poorly in terms of out-of-sample prediction as well
as parameter and regression function estimation performances. The apparent
overfitting of $\ell _{1}$-PQR could be mitigated through employing instead
the adaptive Lasso based approaches of AL-SCAD and AL-MCP or the nonconvex
penalized estimation approaches of QR-SCAD and QR-MCP. Compared to these
four alternative penalized quantile regression approaches, it is worth
noting that both the $\ell _{0}$-PQR and $\ell _{0}$-CQR approaches were
capable of further reducing substantially the incidence of selecting
irrelevant covariates while retaining comparable performance in prediction,
parameter and regression function estimation, and the selection of relevant
covariates. Finally, we also note in this simulation that the FO-based
implementation of $\ell _{0}$-PQR delivered similar pattern of estimated
sparsity to that of the MIO-based $\ell _{0}$-PQR and its predictive and
estimation performances did not fall much behind those of the MIO-based $%
\ell _{0}$-PQR and $\ell _{0}$-CQR approaches. These results also suggest
that our FO-based $\ell _{0}$-PQR method can be a useful standalone approach
to sparse estimation of high dimensional quantile regression models.

\begin{center}
\begin{table}[tbh]
\caption{Simulation comparison for $p=500$ under parameter configuration
(ii) }
\label{Table 3}{\small 
\begin{tabular}{c|cccccccc}
\hline\hline
$p=500$ & \multicolumn{2}{|c}{$\ell _{0}$-PQR} & $\ell _{0}$-CQR & $\ell
_{1} $-PQR & AL-SCAD & AL-MCP & QR-SCAD & QR-MCP \\ 
& MIO & FO &  &  &  &  &  &  \\ \hline
$Num\_irrel$ & 0.22 & 0.22 & 0.6 & 43.87 & 19.07 & 14.12 & 18.89 & 15.78 \\ 
$Avg\_sparsity$ & 8.71 & 8.09 & 9.12 & 53.88 & 28.57 & 23.45 & 28.42 & 25.19
\\ 
$\mathbb{E}\left[ \left\Vert \widehat{\theta }-\theta _{\ast }\right\Vert
_{2}\right] $ & 0.081 & 0.112 & 0.084 & 0.206 & 0.109 & 0.106 & 0.107 & 0.102
\\ 
$\mathbb{E}[|X^{\top }(\widehat{\theta }-\theta _{\ast })|^{2}]$ & 0.007 & 
0.014 & 0.007 & 0.042 & 0.012 & 0.011 & 0.012 & 0.011 \\ 
$in\_RR$ & 1.003 & 1.074 & 0.993 & 0.602 & 0.842 & 0.890 & 0.841 & 0.864 \\ 
$out\_RR$ & 1.132 & 1.219 & 1.137 & 1.517 & 1.206 & 1.196 & 1.199 & 1.186 \\ 
\hline
\end{tabular}%
}
\end{table}
\end{center}

\section{An Application to Conformal Prediction\label{Sec:Application}}

In this section, we compare $\ell _{0}$-PQR and $\ell _{0}$-CQR with $\ell
_{1}$-PQR and the other four alternative penalized quantile regression
approaches (AL-SCAD, AL-MCP, QR-SCAD, QR-MCP) of Section \ref{Sec:Simulation}
via a real data application to conformal prediction of birth weights. In
particular, we employ conformalized quantile regression %
\citep{romano2019conformalized} to construct prediction intervals for birth
weights.

We now describe the split conformalized quantile regression procedure 
\citep[see Algorithm 1
of ][] {romano2019conformalized}. First, we split the data into a proper
training set, indexed by $\mathcal{I}_{1}$, and a calibration set, indexed
by $\mathcal{I}_{2}$. For each quantile regression algorithm, we use the
proper training set $\mathcal{I}_{1}$ to obtain the estimates of two
conditional quantile functions $\widehat{Q}_{\alpha /2}(Y|X=x)$ and $%
\widehat{Q}_{1-\alpha /2}(Y|X=x)$ for a given level $\alpha \in (0,0.5)$.
Then, the following scores are evaluated on the calibration set $\mathcal{I}%
_{2}$ as 
\begin{equation*}
E_{i}\equiv\max \{\widehat{Q}_{\alpha /2}(Y|X=X_{i})-Y_{i},Y_{i}-\widehat{Q}%
_{1-\alpha /2}(Y|X=X_{i})\}
\end{equation*}%
for each $i\in \mathcal{I}_{2}$. Finally, given new covariates $X_{n+1}$,
construct the prediction interval for $Y_{n+1}$ as\ 
\begin{equation}
C(X_{n+1})\equiv\left[ \widehat{Q}_{\alpha /2}(Y|X=X_{n+1})-Q_{1-\alpha }(E,%
\mathcal{I}_{2}),\widehat{Q}_{1-\alpha /2}(Y|X=X_{n+1})+Q_{1-\alpha }(E,%
\mathcal{I}_{2})\right]  \label{prediction interval}
\end{equation}%
where $Q_{1-\alpha }(E,\mathcal{I}_{2})$ is the $(1-\alpha )(1+1/|\mathcal{I}%
_{2}|)$-th empirical quantile of $\{E_{i}:i\in \mathcal{I}_{2}\}$.
Remarkably, Theorem 1 of \citet{romano2019conformalized} guarantees that the
prediction interval \eqref{prediction interval} satisfies the marginal,
distribution-free, finite-sample coverage in the sense that 
\begin{equation*}
\mathbb{P}\left\{ Y_{n+1}\in C(X_{n+1})\right\} \geq 1-\alpha ,
\end{equation*}%
provided that the data $\{(Y_{i},X_{i}):i=1,\ldots ,n+1\}$ are exchangeable.

We look at the dataset on birth weights originally analyzed by %
\citet{almond2005costs}. We use the excerpt from %
\citet{cattaneo2010efficient} available at %
\url{http://www.stata-press.com/data/r13/cattaneo2.dta}. The sample size is
4642 and the outcome of interest is infant birth weight measured in
kilograms. The basic covariates include 20 variables concerning mother's
age, mother's years of education, father's age, father's years of education,
number of prenatal care visits, trimester of first prenatal care visit,
birth order of an infant, months since last birth, an indicator variable
whether a newborn died in previous births, mother's smoking behavior during
pregnancy, mother's alcohol consumption during pregnancy, mother's marital
status, mother's and father's hispanic status and race (being white or not),
an indicator variable whether a mother was born abroad, and three dummy
variables indicating seasons of the birth.

We consider the following five different dictionary specifications. The
first one is concerned with a covariate vector of $p=21$ that includes a
regression intercept together with the aforementioned 20 basic explanatory
variables. The second specification modifies the first by discretizing both
the maternal and paternal years of education into four categories indicating
whether the schooling year is less than 12, exactly 12, between 12 and 16,
or at least 16. In addition, we replace number of prenatal care visits,
months since last birth and both parents' ages by cubic B-spline terms using
4 interior B-spline knots. These allow us to approximate smooth functions of
these variables in the quantile regression analysis. We exclude the B-spline
intercept terms so that the resulting covariate vector for the second
specification has dimension $p=49$. The third covariate specification
consists of all variables in the second specification together with those
obtained by interacting the B-spline expansion terms with the other
explanatory variables. This then renders $p=609$ in the third covariate
specification scenario. Both the fourth and fifth specifications are
constructed using the same procedure as for the third case except that we
enlarge the covariate dimensions by using respectively 12 and 16 interior
B-spline knots for these two high dimensional scenarios, each of which
comprises 1281 and 1617 covariates respectively. Finally, for each covariate
specification, to put the variables on a similar scale, all stochastic
covariates thus constructed are further standardized to have mean zero and
variance unity.

We conduct conformal prediction of infant birth weights with nominal level $%
\alpha =0.1$. We split the sample randomly into four subsets of about equal
size: $\mathcal{I}_{1}$, $\mathcal{I}_{2}$, $\mathcal{I}_{3}$ and $\mathcal{I%
}_{4}$. As described above, the set $\mathcal{I}_{1}$ is the training sample
for the estimation of conditional quantile functions. We perform this
estimation respectively using the $\ell _{0}$-PQR, $\ell _{0}$-CQR, $\ell
_{1}$-PQR, AL-SCAD, AL-MCP, QR-SCAD and QR-MCP approaches. We calibrate the
tuning parameters in these competing quantile regression approaches using
the validation sample $\mathcal{I}_{2}$. Let $\mathcal{S}%
=\{0.1,0.2,...,1.9,2\}\cup \{0.1(0.7)^{s}:s=1,...,8\}$. For penalized
estimation approaches, we calibrate the tuning constants $c_{BC}$ of (\ref%
{L1 PQR tuning parameter}), $c$ of (\ref{L0 penalty tuning parameter}), and $%
\mu $ of (\ref{adaptive Lasso}) and (\ref{fully nonconvex penalized QR})
over the set $\mathcal{S}\cup \{0\}$ for the estimation with the first two
covariate specifications where $p\in \{21,49\}$ and over the set $\mathcal{S}
$ for the high dimensional estimation case where $p\in \{609,1281,1617\}$.
Except for these modifications, implementation of all penalized quantile
regression approaches and estimation and calibration for the $\ell _{0}$-CQR
are performed in the same fashion\ as described in the simulation study.

With the calibrated optimal tuning parameter value, we use the set $\mathcal{%
I}_{3}$ to estimate the out-of-sample quantile prediction risk and
conformalize quantile regression estimates by constructing $\{E_{i}:i\in 
\mathcal{I}_{3}\}$. We then evaluate the coverage performance of the
prediction interval \eqref{prediction
interval} over the test sample $\mathcal{I}_{4}$. We carry out 10
replications of such random splitting exercises and report the averages of
estimated sparsity, out-of-sample prediction risk as well as length and
coverage of the prediction interval across these replications. To mitigate
the computational cost, every MIO computation in this empirical study is
conducted under a 5-minute computational time constraint.

\subsection{Empirical Results}

We summarize in Figures \ref{fig01} and \ref{fig02} statistical performances
under the aforementioned competing quantile regression approaches. For the
basic covariate specification with $p=21$, we also juxtapose and compare the
performance results with those obtained through standard quantile regression
(Std-QR) of \citet{Koenker1978}. As the Std-QR approach does not incur any
tuning parameter, we conduct the Std-QR based conformal prediction by
splitting the dataset evenly into three subsamples of which we use the first
for parameter estimation, the second for estimating the out-of-sample
quantile prediction risk and conformalizing the quantile regression
estimates, and the third for estimating the coverage of the conformalized
prediction interval. We also perform 10 replications of this
sample-splitting procedure and report the average performance results under
the the Std-QR approach. Moreover, for each regularized estimation method,
we report in Figure \ref{fig03} its computational performance, which is
measured by the employed CPU seconds that are averaged over the range of
tuning parameter values and across the random splitting replications. See
also online appendix of the paper for further
details on the variable selection results of this empirical study.

\begin{figure}[htbp]
\caption{Results on Estimated Sparsities and Out-of-Sample Prediction Risks}
\label{fig01}\centering
\medskip \includegraphics[scale=0.43]{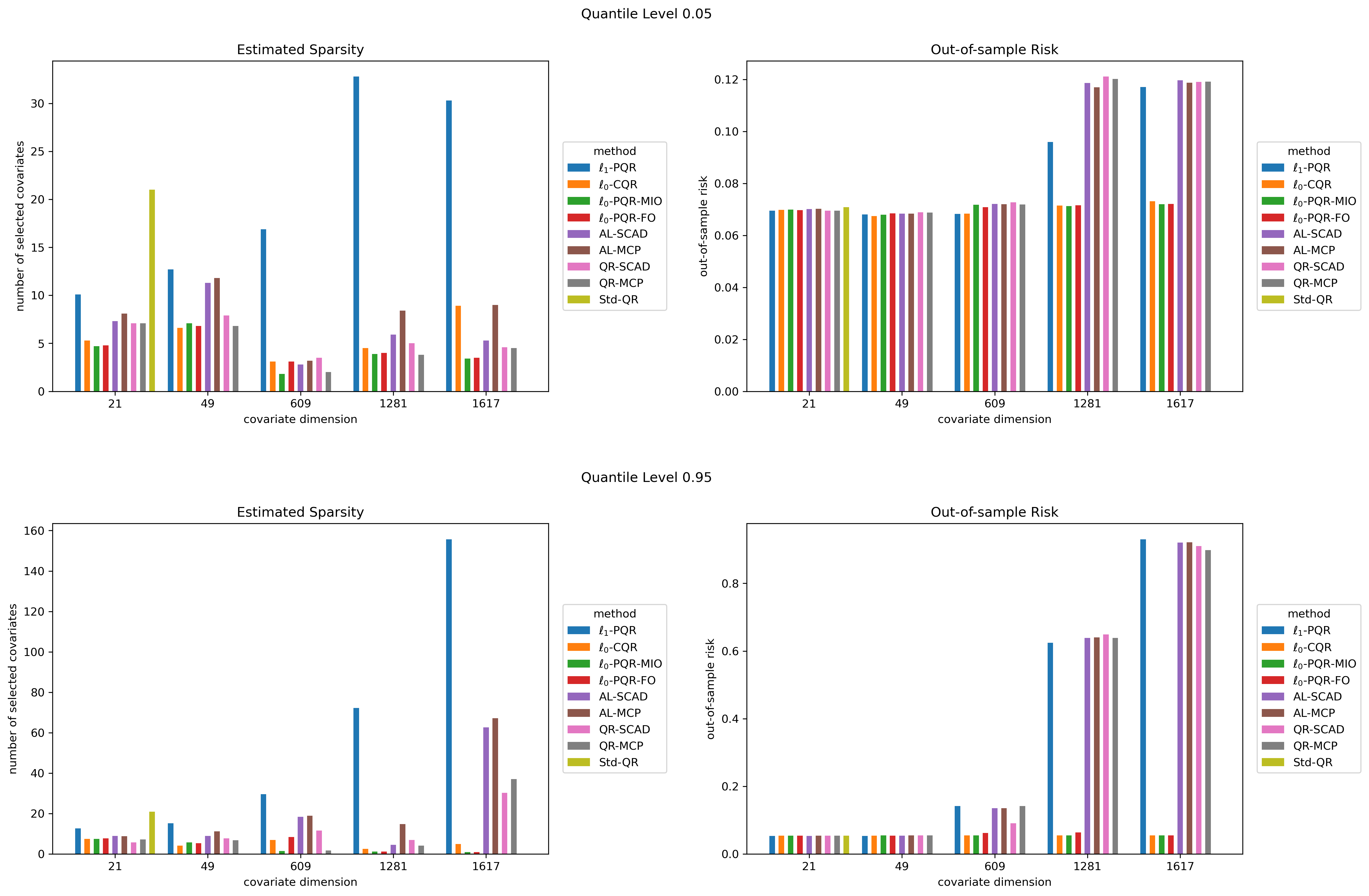}
\end{figure}

\begin{figure}[htbp]
\caption{Lengths and Coverages of Conformalized Prediction Intervals}
\label{fig02}\centering
\medskip \includegraphics[scale=0.43]{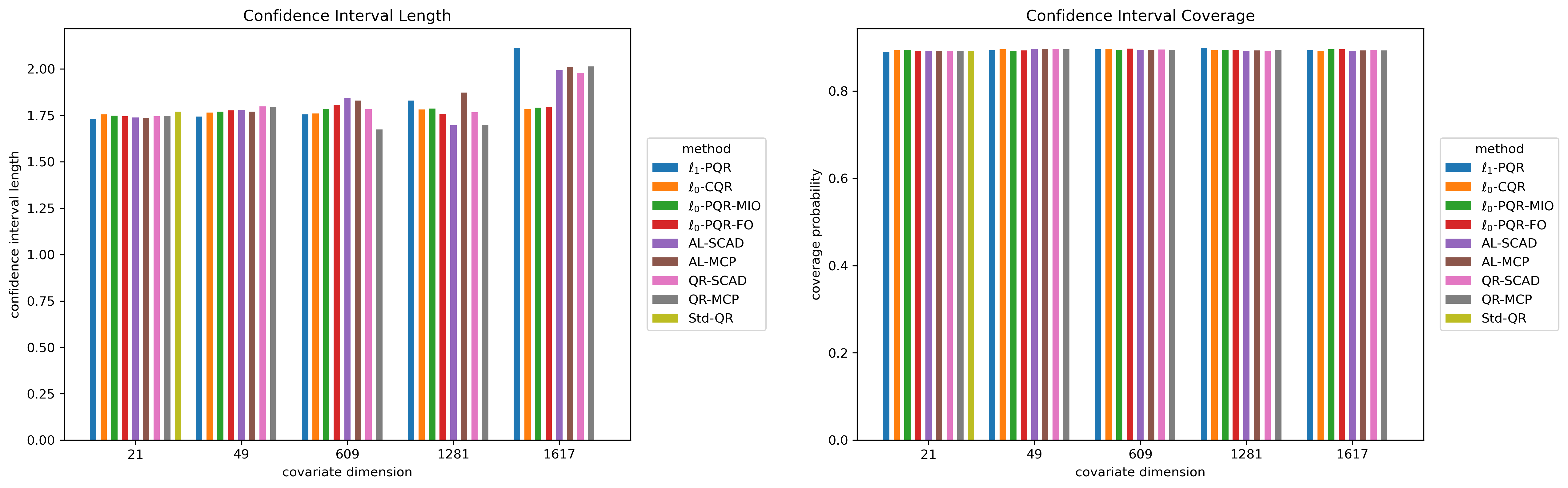}
\end{figure}

From Figure \ref{fig01}, we find that $\ell _{1}$-PQR tended to induce a far
denser estimating model than all the other competing regularized estimation
approaches across all covariate specifications. At the 5\% quantile level,
the average number of selected covariates under $\ell _{1}$-PQR could go
over 30 when $p>10^{3}$, whereas that quantity did not exceed 12 for all the
other high dimensional quantile regression approaches. At the 95\% quantile
level, the average estimated sparsity for $\ell _{1}$-PQR was around 12 when 
$p=21$. Yet, this figure rose quickly as the covariate dimension increased.
It reached around 30, then moved up to 72, and eventually escalated toward
155 as $p$ increased from 609 to 1617. While this excessive sparsity pattern
could be curbed under both our $\ell _{0}$-based approaches and the
approaches of AL-SCAD, AL-MCP, QR-SCAD and QR-MCP, we note that, for the
high dimension scenario with $p=1617$, the average estimated sparsities
under AL-SCAD and AL-MCP still went over 60 and those under QR-SCAD and
QR-MCP were smaller yet remained above 30. By contrast, none of the $\ell
_{0}$-based approaches selected more than 9 variables across all the
covariates specifications.

Concerning the predictive performance, at the lower quantile level, all the
estimation approaches exhibited commensurate out-of-sample quantile
prediction risks in the cases where $p\in\{21,49\}$. For $p=609$, both $\ell
_{0}$-CQR and $\ell _{1}$-PQR performed slightly better than the other
approaches. However, in the high dimensional scenarios where $%
p\in\{1281,1617\}$, both $\ell _{0}$-PQR and $\ell _{0}$-CQR had similar
prediction performances, which clearly dominated those of $\ell _{1}$-PQR
and the other four alternative penalized quantile regression approaches. At
the higher quantile level, all the estimation approaches also performed
comparably well in the low dimensional cases. However, except for the $\ell
_{0}$-based approaches whose out-of-sample prediction risks were all capped
below around 6.5\%, those risks under all the other high dimensional
estimation approaches went over 60\% at $p=1281$ and surged toward 90\% when 
$p$ reached 1617.

For the conformalized prediction intervals, Figure \ref{fig02} indicates
that coverages of these intervals were quite similar across all the
estimation approaches and on average dovetailed well with the nominal size.
However, lengths of the prediction intervals varied across both the
estimation methods and the covariate specifications. For $p\in\{21,49\}$,
all methods performed comparably well though QR-SCAD and QR-MCP appeared to
produce slightly wider prediction intervals at $p=49$. For the case with $%
p=609$, QR-MCP delivered the tightest prediction interval with length 1.67;
those of the other approaches had lengths ranging from 1.75 to 1.84. For $%
p=1281$, the prediction interval under AL-SCAD was the tightest whereas that
under AL-MCP was the widest. The maximal difference in these interval
lengths was about 0.18 in this covariate specification. Yet, this maximal
difference reached 0.33 and hence nearly doubled at $p=1617$ where the
prediction interval lengths under the $\ell _{0}$-based approaches were all
around 1.79 yet those under the other regularized estimation approaches were
at least around 2 and this could exceed 2.11 for the case of $\ell _{1}$%
-PQR. On the whole, the statistical performance results in Figures \ref%
{fig01} and \ref{fig02} reveal that the $\ell _{0}$-based approaches were
capable of delivering sparser solutions than all the other competing
estimation approaches whilst maintaining quite favorable performances in the
prediction accuracy.

\begin{figure}[htbp]
\caption{Results on Computational Performance}
\label{fig03}\centering
\medskip %
\includegraphics[scale=0.45]{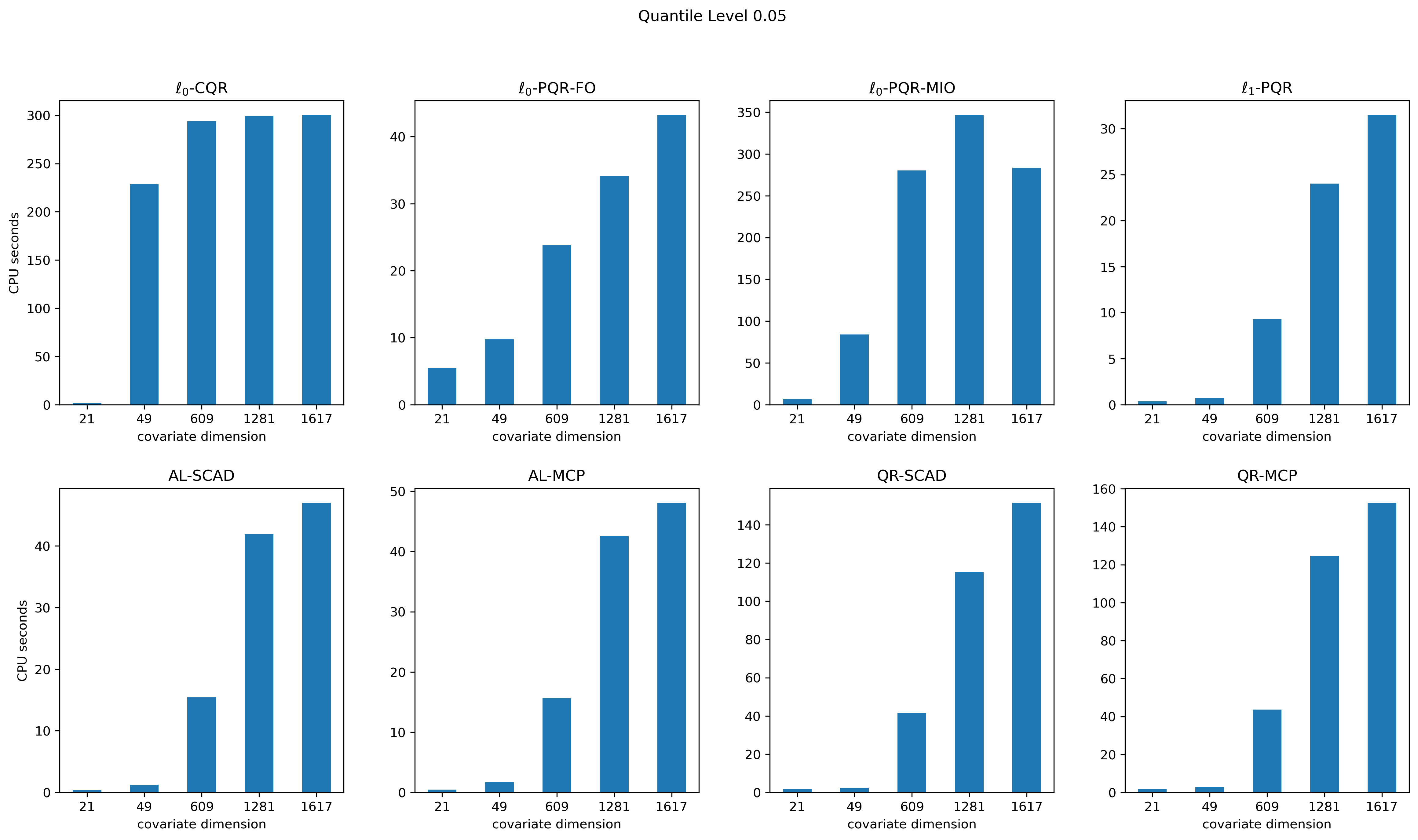}
\includegraphics[scale=0.45]{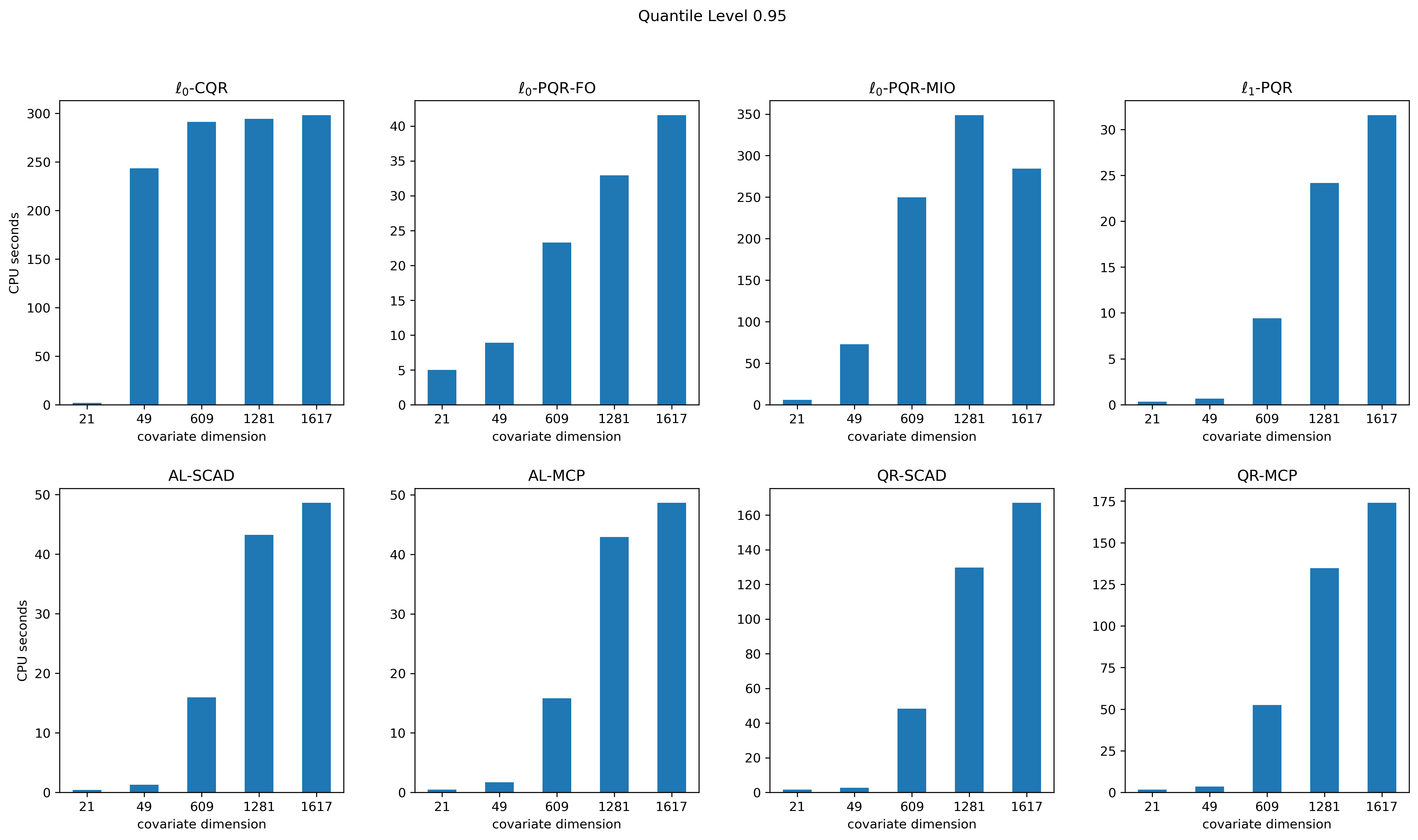}
\end{figure}

We now turn attention to computational performance of the considered
estimation approaches in this study. From Figure \ref{fig03}, it is evident
that $\ell _{1}$-PQR enjoyed the best computational performance with its
average computation time being capped below around 31 CPU seconds across all
the estimation scenarios. FO-based $\ell _{0}$-PQR also performed very well.
It was initially slower than both AL-SCAD and AL-MCP in the low dimensional
estimation cases where its computational time remained below 24 CPU seconds.
Yet, its computational performance became quite competitive relatively to
those of the adaptive Lasso based approaches in the high dimensional cases
with $p\in\{1281,1617\}$ where the average computation time for both QR-SCAD
and QR-MCP could go well over 110 seconds.

Relative to the other approaches, the two methods, $\ell _{0}$-CQR and
MIO-based $\ell _{0}$-PQR, were observed to be far more computationally
intensive. Across all estimation scenarios, Figure \ref{fig03} indicates a
substantial computational performance difference between the MIO and
non-MIO-based approaches. This could be anticipated because of the high
computational complexity in the estimation problems for both $\ell _{0}$-CQR
and MIO-based $\ell _{0}$-PQR. Based on these empirical results, we find
that the FO-based implementation of $\ell _{0}$-PQR could strike a good
balance between statistical and computational performances and thus be a
valuable standalone estimation approach. On the whole, we note that the $%
\ell _{0}$-based approaches could be competitive alternatives to $\ell _{1}$%
-PQR, adaptive Lasso based and the other nonconvex penalized estimation
approaches employed in this numerical study.


\section{Conclusions\label{Sec:Conclusions}}

In this paper, we study estimation of a sparse high dimensional quantile
regression model. The main contributions of this paper are twofold. First,
we derive non-asymptotic expectation bounds on the excess quantile
prediction risk as well as the mean-square parameter and regression function
estimation errors of both the $\ell _{0}$-PQR and $\ell _{0}$-CQR
estimators. These theoretical results imply the near minimax optimal rates
of convergence. Moreover, we characterize expected Hamming loss for the $%
\ell _{0}$-penalized estimator. The second contribution is computational. We
provide an exact computation approach for $\ell _{0}$-PQR through the method
of mixed integer optimization. We also develop a first-order approximation
algorithm for solving large scale $\ell _{0}$-PQR problems. Through Monte
Carlo simulations and a real-data application, we find that both $\ell _{0}$%
-PQR and $\ell _{0}$-CQR perform fairly well and produce much sparser
solutions than $\ell _{1}$-PQR does and also outperform the adaptive Lasso and non-convex penalized quantile regression approaches. Our theoretical and numerical results
suggest that the $\ell _{0}$-based approaches are worthy competitors to 
the $\ell _{1}$-based and non-convex penalized estimation methods in
sparse quantile regression. Recently, \citet{Hazimeh:Mazumder:2020}
developed fast computational methods for $\ell _{0}$-penalized least squares
with an additional $\ell _{1}$- or $\ell _{2}$-penalty term. It is an
interesting future research topic to extend their approach to quantile
regression and investigate its statistical properties.

\section*{Acknowledgements}

We are indebted to the editor, Elie Tamer, an associate editor and two
anonymous referees for constructive comments and suggestions. We would like
to thank Roger Koenker, Rahul Mazumder, Guillaume Pouliot and participants
at 2019 Optimization-Conscious Econometrics Conference in Chicago, 2020
Econometric Society World Congress, and 2021 ASSA Annual Meeting for helpful
comments. We are also grateful to Rahul Mazumder for providing us his code.
This work was supported in part by the Ministry of Science and Technology,
Taiwan (MOST109-2410-H-001-027-MY2), Academia Sinica (AS-CDA-106-H01), the
European Research Council (ERC-2014-CoG-646917-ROMIA), and the UK Economic
and Social Research Council (ESRC) through research grant (ES/P008909/1) to
the CeMMAP.

\appendix%



\section{Proofs\label{Proofs}}

\subsection{Lemmas}

For any $\theta ,\theta _{1},\theta _{2}\in \Theta $, define $\overline{S}%
_{n}(\theta )\equiv S_{n}(\theta )-S(\theta )$, $\Delta (\theta _{1},\theta
_{2})\equiv S(\theta _{1})-S(\theta _{2})$, and $\overline{\Delta }%
_{n}(\theta _{1},\theta _{2})\equiv \overline{S}_{n}(\theta _{1})-\overline{S%
}_{n}(\theta _{2})$. By Assumption \ref{sep-cond}, we have that, for some
given scalar $\varepsilon _{\ast }^{2}$, which will be chosen later, we can
find a point $\theta _{\ast }^{\prime }$ in $\Theta ^{\prime }$ such that $%
\Vert \theta _{\ast }^{\prime }\Vert _{0}=\Vert \theta _{\ast }\Vert _{0}$
and $\Delta (\theta _{\ast }^{\prime },\theta _{\ast })\leq \varepsilon
_{\ast }^{2}$. We start with the following basic inequality.

\begin{lemma}[Basic inequality]
\label{sparsity-basic-new}For $k_{0}\geq s$, 
\begin{equation*}
\Delta (\widehat{\theta },\theta _{\ast }^{\prime })+\lambda \Vert \widehat{%
\theta }-\theta _{\ast }^{\prime }\Vert _{0}\leq \overline{\Delta }%
_{n}(\theta _{\ast }^{\prime },\widehat{\theta })+2\lambda s.
\end{equation*}
\end{lemma}

\begin{proof}[Proof of Lemma~\protect\ref{sparsity-basic-new}]
Using (\ref{penalized QR}), we have that, for $k_{0}\geq s$, 
\begin{equation}
S_{n}(\widehat{\theta })+\lambda \Vert \widehat{\theta }\Vert _{0}\leq
S_{n}(\theta _{\ast }^{\prime })+\lambda \Vert \theta _{\ast }^{\prime
}\Vert _{0}.  \label{basic-ineq-ell0}
\end{equation}%
Using \eqref{basic-ineq-ell0}, we can deduce that 
\begin{equation}
\Delta (\widehat{\theta },\theta _{\ast }^{\prime })+\lambda \Vert \widehat{%
\theta }\Vert _{0}\leq \overline{\Delta }_{n}(\theta _{\ast }^{\prime },%
\widehat{\theta })+\lambda s.  \label{ineq}
\end{equation}%
Then, the lemma follows from (\ref{ineq}) together with an application of
triangle inequality.
\end{proof}

Let 
\begin{equation}
V_{x}\equiv \sup_{\theta \in \mathbb{B}(k_{0})}\frac{\overline{\Delta }%
_{n}(\theta _{\ast }^{\prime },\theta )}{\Delta (\theta ,\theta _{\ast
})+\varepsilon _{\ast }^{2}+x^{2}}.  \label{Vx}
\end{equation}

\begin{lemma}[Preliminary Probability Bounds]
\label{lemma-prob-bound} Let Assumptions \ref{iden-cond} and \ref{sep-cond}
hold. Suppose $k_{0}\geq s$. Then for a constant $0<\eta <1$, 
\begin{align}
& \mathbb{P}\left[ \Delta (\widehat{\theta },\theta _{\ast })\geq \frac{2}{%
1-\eta }\lambda s+\frac{1+\eta }{1-\eta }\varepsilon _{\ast }^{2}+\frac{\eta 
}{1-\eta }x^{2}\right] \leq \mathbb{P}\left( V_{x}\geq \eta \right) ,
\label{a3} \\
& \mathbb{P}\left[ \Vert \widehat{\theta }-\theta _{\ast }\Vert _{0}\geq 
\frac{4-2\eta }{1-\eta }s+\lambda ^{-1}\left\{ \frac{1+\eta }{1-\eta }%
\varepsilon _{\ast }^{2}+\frac{\eta }{1-\eta }x^{2}\right\} \right] \leq 
\mathbb{P}\left( V_{x}\geq \eta \right) .  \label{a3-hamming-loss}
\end{align}
\end{lemma}

\begin{proof}[Proof of Lemma~\protect\ref{lemma-prob-bound}]
First, since $\lambda \Vert \widehat{\theta }-\theta _{\ast }^{\prime }\Vert
_{0}$ is always non-negative, it follows from Lemma~\ref{sparsity-basic-new}
that 
\begin{equation}
\Delta (\widehat{\theta },\theta _{\ast })=\Delta (\widehat{\theta },\theta
_{\ast }^{\prime })+\Delta (\theta _{\ast }^{\prime },\theta _{\ast })\leq
\varepsilon _{\ast }^{2}+2\lambda s+\overline{\Delta }_{n}(\theta _{\ast
}^{\prime },\widehat{\theta }).  \label{a2-0}
\end{equation}%
Using (\ref{Vx}) and (\ref{a2-0}), we have that, if $V_{x}<\eta $ for some
positive constant $\eta <1$, then%
\begin{align}
\Delta (\widehat{\theta },\theta _{\ast })& \leq 2\lambda s+\varepsilon
_{\ast }^{2}+V_{x}\left[ \Delta (\widehat{\theta },\theta _{\ast
})+\varepsilon _{\ast }^{2}+x^{2}\right]  \notag \\
& <\frac{2}{1-\eta }\lambda s+\frac{1+\eta }{1-\eta }\varepsilon _{\ast
}^{2}+\frac{\eta }{1-\eta }x^{2}.  \label{lem-Vx-eq-ref}
\end{align}%
Furthermore, note that 
\begin{equation*}
\Delta (\widehat{\theta },\theta _{\ast })+\lambda \Vert \widehat{\theta }%
-\theta _{\ast }^{\prime }\Vert _{0}\leq \varepsilon _{\ast }^{2}+\Delta (%
\widehat{\theta },\theta _{\ast }^{\prime })+\lambda \Vert \widehat{\theta }%
-\theta _{\ast }^{\prime }\Vert _{0}.
\end{equation*}

By Assumption \ref{iden-cond}, $\Delta (\widehat{\theta },\theta _{\ast
})\geq 0$. Therefore, using Lemma~\ref{sparsity-basic-new}, it follows that $%
\Vert \widehat{\theta }-\theta _{\ast }^{\prime }\Vert _{0}\leq 2s+\lambda
^{-1}[\overline{\Delta }_{n}(\theta _{\ast }^{\prime },\widehat{\theta }%
)+\varepsilon _{\ast }^{2}].$ Thus, if $V_{x}<\eta $, we have 
\begin{equation*}
\begin{split}
\Vert \widehat{\theta }-\theta _{\ast }^{\prime }\Vert _{0}& \leq 2s+\lambda
^{-1}V_{x}\left[ \Delta (\widehat{\theta },\theta _{\ast })+\varepsilon
_{\ast }^{2}+x^{2}\right] +\lambda ^{-1}\varepsilon _{\ast }^{2} \\
& <2s+\lambda ^{-1}\eta \left[ \frac{2}{1-\eta }\lambda s+\frac{1+\eta }{%
1-\eta }\varepsilon _{\ast }^{2}+\frac{\eta }{1-\eta }x^{2}+\varepsilon
_{\ast }^{2}+x^{2}\right] +\lambda ^{-1}\varepsilon _{\ast }^{2}\;\;\text{by %
\eqref{lem-Vx-eq-ref}.}
\end{split}%
\end{equation*}%
Arranging the terms, we therefore have that 
\begin{equation*}
\mathbb{P}\left[ \Vert \widehat{\theta }-\theta _{\ast }^{\prime }\Vert
_{0}\geq \frac{2}{1-\eta }s+\lambda ^{-1}\left\{ \frac{1+\eta }{1-\eta }%
\varepsilon _{\ast }^{2}+\frac{\eta }{1-\eta }x^{2}\right\} \right] \leq 
\mathbb{P}\left( V_{x}\geq \eta \right) ,
\end{equation*}%
which yields (\ref{a3-hamming-loss}) by an application of triangle
inequality.
\end{proof}

For $q\geq 0$, let $\mathbb{B}^{\prime
}(q)\equiv \{\theta \in \Theta ^{\prime }:\Vert \theta \Vert _{0}\leq q\}$. By Assumption~\ref{sep-cond}, 
\begin{equation*}
V_{x}=\sup_{\theta \in \mathbb{B}^{\prime }(k_{0})}\frac{\overline{\Delta }%
_{n}(\theta _{\ast }^{\prime },\theta )}{\Delta (\theta ,\theta _{\ast
})+\varepsilon _{\ast }^{2}+x^{2}}.
\end{equation*}%
In other words, it suffices to take the supremum over the countable subset $%
\mathbb{B}^{\prime }(k_{0})$.

To bound $\mathbb{P}(V_{x}\geq \eta )$ in (\ref{a3}) and (\ref%
{a3-hamming-loss}), we will use Bousquet's inequality and a technical lemma
from \citet{massart2006}. For the sake of easy referencing, these results
are reproduced below.

\begin{lemma}[Bousquet's inequality]
\label{lemma-Bousquet} Suppose that $\mathcal{F}$ is a countable family of
measurable functions such that for every $f\in \mathcal{F},P(f^{2})\leq v$
and $\Vert f\Vert _{\infty }\leq b$ for some positive constants $v$ and $b$.
Define $Z\equiv \sup_{f\in \mathcal{F}}(P_{n}-P)(f)$. Then for every  $y>0$, 
\begin{equation}
\mathbb{P}\left[ Z-\mathbb{E}\left[ Z\right] \geq \sqrt{2\frac{\left( v+4b%
\mathbb{E}[Z]\right) y}{n}}+\frac{2by}{3n}\right] \leq e^{-y}.
\label{Bousquet inequality}
\end{equation}
\end{lemma}

\begin{lemma}[Lemma A.5 of Massart and N\'{e}d\'{e}lec (2006)]
\label{lemma-weighted-emp-process} Let $\mathcal{S}$ be a countable set, $%
u\in \mathcal{S}$ and $a:\mathcal{S}\rightarrow \mathbb{R}_{+}$ such that $%
a(u)=\inf_{t\in \mathcal{S}}a(t)$. Define $\mathcal{B}(\varepsilon )=\{t\in 
\mathcal{S}:a(t)\leq \varepsilon \}$. Let $Z$ be a process indexed by $%
\mathcal{S}$ and assume that the nonnegative random variable $\sup_{t\in 
\mathcal{B}(\varepsilon )}[Z(u)-Z(t)]$ has finite expectation for any
positive number $\varepsilon $. Let $\psi $ be a nonnegative function on $%
\mathbb{R}_{+}$ such that $\psi (x)/x$ is nonincreasing on $\mathbb{R}_{+}$
and satisfies for some positive $\varepsilon _{\ast }:$ 
\begin{equation*}
\mathbb{E}\left[ \sup_{t\in \mathcal{B}(\varepsilon )}[Z(u)-Z(t)]\right]
\leq \psi (\varepsilon )\ \ \ \text{for any $\varepsilon \geq \varepsilon
_{\ast }$}.
\end{equation*}%
Then, one has, for any positive number $x\geq \varepsilon _{\ast }$, 
\begin{equation*}
\mathbb{E}\left[ \sup_{t\in \mathcal{S}}\left( \frac{Z(u)-Z(t)}{%
a^{2}(t)+x^{2}}\right) \right] \leq 4x^{-2}\psi (x).
\end{equation*}
\end{lemma}

\subsection{Proof of Theorem \protect\ref{thm-main-1}}

\begin{proof}[Proof of Theorem \protect\ref{thm-main-1}]
We prove Theorem \ref{thm-main-1} by adopting the ideas behind the proof of
Theorem 2 in \citet{massart2006}. By Assumption \ref{iden-cond}, $\Delta (\theta ,\theta _{\ast })\geq 0$ for
all $\theta \in \mathbb{B}^{\prime }(k_{0})$. Using this fact and
Assumptions \ref{L-cond} and \ref{design-cond-1}, we have that, for all $%
i=1,\ldots ,n$ and for all $\theta \in \mathbb{B}^{\prime }(k_{0})$,%
\begin{equation*}
\left\vert \frac{\rho (Y_{i},X_{i}^{\top }\theta _{\ast }^{\prime })-\rho
(Y_{i},X_{i}^{\top }\theta )}{\Delta (\theta ,\theta _{\ast })+\varepsilon
_{\ast }^{2}+x^{2}}\right\vert \leq \frac{LB^{2}\left( s+k_{0}\right) }{%
x^{2}}\equiv b_{x}.
\end{equation*}%
Moreover, by Assumptions \ref{L-cond}, \ref{design-cond-1}, \ref{key-cond}
and \ref{design-cond-2}, for all $\theta \in \mathbb{B}^{\prime }(k_{0})$,%
\begin{eqnarray*}
\mathbb{E}\left( \left[ \rho (Y_{i},X_{i}^{\top }\theta _{\ast }^{\prime
})-\rho (Y_{i},X_{i}^{\top }\theta )\right] ^{2}\right) &\leq
&L^{2}B^{2}\left\Vert \theta -\theta _{\ast }^{\prime }\right\Vert _{1}^{2}
\\
&\leq &L^{2}B^{2}\Vert \theta -\theta _{\ast }^{\prime }\Vert _{0}\left(
\Vert \theta -\theta _{\ast }\Vert _{2}+\Vert \theta _{\ast }^{\prime
}-\theta _{\ast }\Vert _{2}\right) ^{2} \\
&\leq &L^{2}B^{2}\left( s+k_{0}\right) \kappa _{1}^{-2}\kappa
_{0}^{-2}\left( \sqrt{\Delta (\theta ,\theta _{\ast })}+\varepsilon _{\ast
}\right) ^{2} \\
&\leq &2L^{2}B^{2}\left( s+k_{0}\right) \kappa _{1}^{-2}\kappa
_{0}^{-2}\left( \Delta (\theta ,\theta _{\ast })+\varepsilon _{\ast
}^{2}\right)
\end{eqnarray*}%
and therefore%
\begin{eqnarray*}
\mathbb{E}\left( \left[ \frac{\rho (Y_{i},X_{i}^{\top }\theta _{\ast
}^{\prime })-\rho (Y_{i},X_{i}^{\top }\theta )}{\Delta (\theta ,\theta
_{\ast })+\varepsilon _{\ast }^{2}+x^{2}}\right] ^{2}\right) &\leq &\frac{%
2L^{2}B^{2}\left( s+k_{0}\right) \kappa _{1}^{-2}\kappa _{0}^{-2}\left(
\Delta (\theta ,\theta _{\ast })+\varepsilon _{\ast }^{2}\right) }{\left[
\Delta (\theta ,\theta _{\ast })+\varepsilon _{\ast }^{2}+x^{2}\right] ^{2}}
\\
&\leq &\frac{2L^{2}B^{2}\left( s+k_{0}\right) \kappa _{1}^{-2}\kappa
_{0}^{-2}}{x^{2}}\sup_{\varepsilon \geq 0}\frac{\varepsilon }{\varepsilon
+x^{2}} \\
&\leq &\frac{2L^{2}B^{2}\left( s+k_{0}\right) \kappa _{1}^{-2}\kappa
_{0}^{-2}}{x^{2}}\equiv v_{x}.
\end{eqnarray*}%
Choose $b=b_{x}$ and $v=v_{x}$ in Lemma~\ref{lemma-Bousquet}. By (\ref%
{Bousquet inequality}), we then have that for every positive $y$, 
\begin{equation}
\mathbb{P}\left[ V_{x}-\mathbb{E}[V_{x}]\geq \sqrt{2\frac{(v_{x}+4b_{x}%
\mathbb{E}[V_{x}])y}{n}}+\frac{2b_{x}y}{3n}\right] \leq \exp (-y).
\label{Bousquet-ineq}
\end{equation}%
We now bound $\mathbb{E}[V_{x}]$ using Lemma~\ref{lemma-weighted-emp-process}%
. Let $a^{2}(\theta )\equiv \Delta (\theta _{\ast }^{\prime },\theta _{\ast
})\vee \Delta (\theta ,\theta _{\ast })$ for any $\theta \in \Theta $. Then $%
\Delta (\theta ,\theta _{\ast })\leq a^{2}(\theta )\leq \Delta (\theta
,\theta _{\ast })+\varepsilon _{\ast }^{2}$. Therefore, 
\begin{equation}
\mathbb{E}[V_{x}]\leq \mathbb{E}\left[ \sup_{\theta \in \mathbb{B}^{\prime
}(k_{0})}\frac{\overline{\Delta }_{n}(\theta _{\ast }^{\prime },\theta )}{%
a^{2}(\theta )+x^{2}}\right]
\end{equation}%
and for every $\varepsilon \geq \varepsilon _{\ast }$, 
\begin{equation*}
\mathbb{E}\left[ \sup_{\theta \in \mathbb{B}^{\prime }(k_{0}):a(\theta )\leq
\varepsilon }\overline{\Delta }_{n}(\theta _{\ast }^{\prime },\theta )\right]
\leq \mathbb{E}\left[ \sup_{\theta \in \mathbb{B}^{\prime }(k_{0}):\Delta
(\theta ,\theta _{\ast })\leq \varepsilon ^{2}}\overline{\Delta }_{n}(\theta
_{\ast }^{\prime },\theta )\right] .
\end{equation*}%
The next step is to find a function $\psi $ such that 
\begin{equation}
\mathbb{E}\left[ \sup_{\theta \in \mathbb{B}^{\prime }(k_{0}):\Delta (\theta
,\theta _{\ast })\leq \varepsilon ^{2}}\overline{\Delta }_{n}(\theta _{\ast
}^{\prime },\theta )\right] \leq \psi (\varepsilon )\ \ \ \text{for any $%
\varepsilon \geq \varepsilon _{\ast }$}.
\end{equation}%
Let $\epsilon _{1},\ldots ,\epsilon _{n}$ denote a Rademacher sequence that
is independent of $\{(Y_{i},X_{i}):i=1,\ldots ,n\}$. By the symmetrization
and contraction theorems (e.g., Theorems 14.3 and 14.4 of %
\citet{buhlmann2011}), 
\begin{align*}
& \mathbb{E}\left[ \sup_{\theta \in \mathbb{B}^{\prime }(k_{0}):\Delta
(\theta ,\theta _{\ast })\leq \varepsilon ^{2}}\overline{\Delta }_{n}(\theta
_{\ast }^{\prime },\theta )\right] \\
& \leq \mathbb{E}\left[ \sup_{\theta \in \mathbb{B}^{\prime }(k_{0}):\Delta
(\theta ,\theta _{\ast })\leq \varepsilon ^{2}}|\overline{\Delta }%
_{n}(\theta _{\ast }^{\prime },\theta )|\right] \\
& \leq 2\mathbb{E}\left[ \sup_{\theta \in \mathbb{B}^{\prime }(k_{0}):\Delta
(\theta ,\theta _{\ast })\leq \varepsilon ^{2}}\left\vert
n^{-1}\sum_{i=1}^{n}\epsilon _{i}\left\{ \rho (Y_{i},X_{i}^{\top }\theta
)-\rho (Y_{i},X_{i}^{\top }\theta _{\ast }^{\prime })\right\} \right\vert %
\right] \\
& \leq 4L\mathbb{E}\left[ \sup_{\theta \in \mathbb{B}^{\prime
}(k_{0}):\Delta (\theta ,\theta _{\ast })\leq \varepsilon ^{2}}\left\vert
n^{-1}\sum_{i=1}^{n}\epsilon _{i}X_{i}^{\top }(\theta -\theta _{\ast
}^{\prime })\right\vert \right] .
\end{align*}%
By H\"{o}lder's inequality%
\begin{equation*}
\left\vert n^{-1}\sum_{i=1}^{n}\epsilon _{i}X_{i}^{\top }(\theta -\theta
_{\ast }^{\prime })\right\vert \leq \Vert \theta -\theta _{\ast }^{\prime
}\Vert _{1}\max_{1\leq j\leq p}\left\vert n^{-1}\sum_{i=1}^{n}\epsilon
_{i}X_{i}^{(j)}\right\vert ,
\end{equation*}%
where $X_{i}^{(j)}$ denotes the $j$-th component of the covariate vector $%
X_{i}$.

For $\theta \in \mathbb{B}^{\prime }(k_{0})$ that satisfies $\Delta (\theta
,\theta _{\ast })\leq \varepsilon ^{2}$, we have that, by Assumptions \ref%
{key-cond} and \ref{design-cond-2},%
\begin{eqnarray*}
\Vert \theta -\theta _{\ast }^{\prime }\Vert _{1} &\leq &\sqrt{\Vert \theta
-\theta _{\ast }^{\prime }\Vert _{0}}\left( \Vert \theta -\theta _{\ast
}\Vert _{2}+\Vert \theta _{\ast }^{\prime }-\theta _{\ast }\Vert _{2}\right)
\\
&\leq &\left( s+k_{0}\right) ^{1/2}\kappa _{1}^{-1}\kappa _{0}^{-1}\left( 
\sqrt{\Delta \left( \theta ,\theta _{\ast }\right) }+\sqrt{\Delta \left(
\theta _{\ast }^{\prime },\theta _{\ast }\right) }\right) \\
&\leq &\left( s+k_{0}\right) ^{1/2}\kappa _{1}^{-1}\kappa _{0}^{-1}\left(
\varepsilon +\varepsilon _{\ast }\right).
\end{eqnarray*}%
Therefore, we have that for any $\varepsilon \geq \varepsilon _{\ast }$, 
\begin{align*}
\mathbb{E}\left[ \sup_{\theta \in \mathbb{B}^{\prime }(k_{0}):\Delta (\theta
,\theta _{\ast })\leq \varepsilon ^{2}}\overline{\Delta }_{n}(\theta _{\ast
}^{\prime },\theta )\right] & \leq 8L\left( s+k_{0}\right) ^{1/2}\kappa
_{1}^{-1}\kappa _{0}^{-1}\varepsilon \mathbb{E}\left[ \max_{1\leq j\leq
p}\left\vert n^{-1}\sum_{i=1}^{n}\epsilon _{i}X_{i}^{(j)}\right\vert \right]
\\
& \leq C\left( s+k_{0}\right) ^{1/2}\varepsilon \sqrt{\frac{2\ln (2p)}{n}},
\end{align*}%
where the last inequality follows from Hoeffding's inequality (e.g., Lemma
14.14 of \citet{buhlmann2011}) together with (\ref{C0}) and Assumption \ref%
{design-cond-1}. Hence, we set 
\begin{equation}
\psi (x)\equiv C\left( s+k_{0}\right) ^{1/2}x\sqrt{{2\ln (2p)}/{n}}.
\label{form of psi}
\end{equation}%
Thus, by Lemma~\ref{lemma-weighted-emp-process}, for any $x\geq \varepsilon
_{\ast }$, 
\begin{equation*}
\mathbb{E}[V_{x}]\leq \frac{4C}{x}(s+k_{0})^{1/2}\sqrt{\frac{2\ln (2p)}{n}}.
\end{equation*}%
For every $y\geq 1$, set $x=\sqrt{My}\varepsilon _{\ast }$ for some constant 
$M\geq 1$, which will be chosen below, and 
\begin{equation}
\varepsilon _{\ast }=4C(s+k_{0})^{1/2}\sqrt{\frac{2\ln (2p)}{n}}.
\label{epsilon_star}
\end{equation}%
By (\ref{C0}) and (\ref{condition on p}), we have that $C^{2}\ln (2p)\geq
LB^{2}$ and $\ln (2p)\geq 1$. Therefore,%
\begin{align*}
\mathbb{E}[V_{x}]& \leq \frac{1}{\sqrt{My}}\leq \frac{1}{M^{1/2}}, \\
\frac{b_{x}y}{n}& =\frac{LB^{2}}{32MC^{2}\ln (2p)}\leq \frac{1}{32M}, \\
\frac{v_{x}y}{n}& =\frac{1}{1024M\ln (2p)}\leq \frac{1}{1024M},
\end{align*}%
which implies that%
\begin{eqnarray*}
&&\mathbb{E}[V_{x}]+\sqrt{2\frac{(v_{x}+4b_{x}\mathbb{E}[V_{x}])y}{n}}+\frac{%
2b_{x}y}{3n} \\
&\leq &\frac{1}{M^{1/2}}+\sqrt{\frac{1}{512M}+\frac{1}{4M^{3/2}}}+\frac{1}{48M%
} \\
&\leq &\frac{3}{M^{1/2}}.
\end{eqnarray*}%
It thus follows from (\ref{Bousquet-ineq}) that%
\begin{equation}
\mathbb{P}\left[ V_{x}\geq 3M^{-1/2}\right] \leq \exp (-y)
\label{probability inequality for Vx}
\end{equation}%
for $y\geq 1$. Now choose a sufficiently large $M$ that satisfies 
\begin{equation}
\eta \geq 3M^{-1/2}  \label{condition on M}
\end{equation}%
for any given positive constant $\eta <1$. Putting together (\ref{C0}), (\ref%
{a3}), (\ref{a3-hamming-loss}), (\ref{epsilon_star}), (\ref{probability
inequality for Vx}) and (\ref{condition on M}), we then have that%
\begin{align*}
& \mathbb{P}\left[ \Delta (\widehat{\theta },\theta _{\ast })\geq \frac{2}{%
1-\eta }\lambda s+32C^{2}(s+k_{0})\left( \frac{1+\eta +M\eta y}{1-\eta }%
\right) \frac{\ln (2p)}{n}\right] \leq \exp (-y), \\
& \mathbb{P}\left[ \Vert \widehat{\theta }-\theta _{\ast }\Vert _{0}\geq 
\frac{4-2\eta }{1-\eta }s+32\lambda ^{-1}C^{2}(s+k_{0})\left( \frac{1+\eta
+M\eta y}{1-\eta }\right) \frac{\ln (2p)}{n}\right] \leq \exp (-y)
\end{align*}%
for $y\geq 1$.
\end{proof}

\subsection{Proof of Theorem \protect\ref{mean excess risk and mean squared
estimation error}}

\begin{proof}[Proof of Theorem \protect\ref{mean excess risk and mean
squared estimation error}]
By (\ref{probability bound}) of Theorem \ref{thm-main-1} with the choice of $%
\eta =1/2$, we have that, for every $y\geq 1$,%
\begin{equation}
\mathbb{P}\left[ S(\widehat{\theta })-S\left( \theta _{\ast }\right) \geq
A+By\right] \leq \exp (-y),  \label{tail probability inequality}
\end{equation}%
where%
\begin{eqnarray*}
A &\equiv &4\lambda s+96C^{2}(s+k_{0})\frac{\ln (2p)}{n}, \\
B &\equiv &32MC^{2}(s+k_{0})\frac{\ln (2p)}{n}.
\end{eqnarray*}%
and the constant $C$ is given by (\ref{C0}). Since $S(\widehat{\theta })\geq
S\left( \theta _{\ast }\right) $, result (\ref{mean excess risk}) thus
follows by noting that 
\begin{align*}
& \mathbb{E}\left[ S(\widehat{\theta })-S\left( \theta _{\ast }\right) %
\right] \\
& =\int_{0}^{\infty }\mathbb{P}\left[ S(\widehat{\theta })-S\left( \theta
_{\ast }\right) \geq t\right] dt \\
& =B\int_{-A/B}^{\infty }\mathbb{P}\left[ S(\widehat{\theta })-S\left(
\theta _{\ast }\right) \geq A+By\right] dy \\
& =B\int_{-A/B}^{1}\mathbb{P}\left[ S(\widehat{\theta })-S\left( \theta
_{\ast }\right) \geq A+By\right] dy+B\int_{1}^{\infty }\mathbb{P}\left[ S(%
\widehat{\theta })-S\left( \theta _{\ast }\right) \geq A+By\right] dy \\
& \leq A+2B,
\end{align*}%
where the last inequality above follows from applying (\ref{tail probability
inequality}) for $y\in \lbrack 1,\infty )$. Moreover, using (\ref{margin
condition}) and (\ref{sparse-eigen}), we can deduce that%
\begin{eqnarray*}
\mathbb{E}\left[ R(\widehat{\theta })\right] &\leq &\kappa _{0}^{-2}\mathbb{E%
}\left[ S(\widehat{\theta })-S\left( \theta _{\ast }\right) \right] , \\
\mathbb{E}\left[ \left\Vert \widehat{\theta }-\theta _{\ast }\right\Vert
_{2}^{2}\right] &\leq &\kappa _{1}^{-2}\kappa _{0}^{-2}\mathbb{E}\left[ S(%
\widehat{\theta })-S\left( \theta _{\ast }\right) \right] ,
\end{eqnarray*}%
which, given (\ref{mean excess risk}), therefore imply (\ref{mean squared
prediction error}) and (\ref{mean squared estimation error}).
\end{proof}

\subsection{Proof of Theorem \protect\ref{thm:sparsity}}

\begin{proof}[Proof of Theorem \protect\ref{thm:sparsity}]
By (\ref{condition on p}), $p\geq 2$ so that $\ln (2p)\leq 2\ln p$. Thus
using \eqref{probability bound hamming} in Theorem \ref{thm-main-1}, we have
that 
\begin{equation*}
\mathbb{P}\left[ \Vert \widehat{\theta }-\theta _{\ast }\Vert _{0}\geq \frac{%
4-2\eta }{1-\eta }s+64\lambda ^{-1}C^{2}(s+k_{0})\left( \frac{1+\eta +M\eta y%
}{1-\eta }\right) \frac{\ln (p)}{n}\right] \leq \exp (-y)
\end{equation*}%
for $y\geq 1$. Choose the smallest $M$ that satisfies \eqref{condition on M}%
, i.e., $\eta =3M^{-1/2}$. Because it is assumed that $k_{0}\leq C_{k}s$ for
a fixed constant $C_{k}$ and $\lambda =C_{\lambda }{\ln p}/{n}$, taking $%
C_{\lambda }=64\zeta _{\lambda }C^{2}(C_{k}+1)$ for some constant $\zeta
_{\lambda }\geq 1$, we have that for $y\geq 1$, 
\begin{equation}
\mathbb{P}\left[ \Vert \widehat{\theta }-\theta _{\ast }\Vert _{0}\geq As+Bsy%
\right] \leq \exp (-y),  \label{key-step-derived-sparsity-more}
\end{equation}%
where 
\begin{equation*}
A=\frac{4-2\eta }{1-\eta }+\zeta _{\lambda }^{-1}\frac{1+\eta }{1-\eta }\ \
\ \text{ and }\ \ \ B=\zeta _{\lambda }^{-1}\frac{9}{\eta (1-\eta )}.
\end{equation*}%
Using the integrated tail probability expectation formula for nonnegative
integer valued random variables (see e.g.\thinspace \citet{lo2019}) and
following similar steps in the Proof of Theorem \ref{mean excess risk and
mean squared estimation error}, we have that 
\begin{equation*}
\mathbb{E}\left[ \Vert \widehat{\theta }-\theta _{\ast }\Vert _{0}\right]
\leq (A+2B)s.
\end{equation*}%
The conclusion of the theorem follows by first choosing a sufficiently small 
$\eta $ and then selecting a sufficiently large $\zeta _{\lambda }$.
\end{proof}

\subsection{Proof of Corollary \protect\ref{thm-cor-1}}

\begin{proof}[Proof of Corollary \protect\ref{thm-cor-1}]
As discussed in Section \ref{theory:assump}, Assumptions \ref{iden-cond}--%
\ref{design-cond-2} are satisfied for quantile regression with the Lipschitz
constant $L=1$, Assumption \ref{sep-cond} holds by \eqref{boundedness} and
the presumption on the finiteness of $\mathbb{E}|Y|$, and Assumption \ref%
{design-cond-2} holds with $\kappa _{1}=\sqrt{\omega }$.

Note that, for any $\theta \in \mathbb{B}(k_{0})$, we have that $\left\vert
x^{\top }(\theta -\theta _{\ast })\right\vert \leq B^{2}\left(
k_{0}+s\right) $ by \eqref{boundedness}. Using (\ref{difference in QR
objective functions}), it hence follows from assumption (iv) of this
corollary that%
\begin{equation*}
S(\theta )-S(\theta _{\ast })\geq \frac{c_{u}}{2}\mathbb{E}\left[ \left\vert
X^{\top }(\theta -\theta _{\ast })\right\vert ^{2}\right] \text{ for all }%
\theta \in \mathbb{B}(k_{0}).
\end{equation*}%
Thus, Assumption \ref{key-cond} of Theorem \ref{thm-main-1} also holds with $%
\kappa _{0}=\sqrt{c_{u}/2}$. As a result, the corollary follows from Theorem %
\ref{mean excess risk and mean squared estimation error}.
\end{proof}

\subsection{Proof of Corollary \protect\ref{thm-cor-2}}

\begin{proof}[Proof of Corollary \protect\ref{thm-cor-2}]
Repeating arguments used in the proof of Theorem \ref{thm-main-1} with $%
\lambda =0$, we have that 
\begin{equation*}
\mathbb{P}\left[ \Delta (\tilde{\theta},\theta _{\ast })\geq 32C^{2}\left(
s+q\right) \left( \frac{1+\eta +M\eta y}{1-\eta }\right) \frac{\ln (2p)}{n}%
\right] \leq \exp (-y)
\end{equation*}%
for $y\geq 1$. Then we can proceed as in the proof of Corollary \ref%
{thm-cor-1}.
\end{proof}

\subsection{Proof of Proposition \protect\ref{stationary point}}

\begin{proof}[Proof of Proposition \protect\ref{stationary point}]
(a) Let $\widehat{t}$ denote a point in $H_{\delta ,l}(\widehat{\theta }%
_{\delta })$. By (\ref{H(t)}) and (\ref{H(t) and Q(b)}), 
\begin{equation}
\widetilde{Q}_{n}(\widehat{t};\widehat{\theta }_{\delta },\delta ,l)\leq 
\widetilde{Q}_{n}(\widehat{\theta }_{\delta };\widehat{\theta }_{\delta
},\delta ,l).  \label{a}
\end{equation}%
Using (\ref{quadratic envelope}), we have that%
\begin{equation}
\widetilde{Q}_{n}(\widehat{\theta }_{\delta };\widehat{\theta }_{\delta
},\delta ,l)=Q_{n}(\widehat{\theta }_{\delta };\delta )\leq Q_{n}(\widehat{t}%
;\delta )\leq \widetilde{Q}_{n}(\widehat{t};\widehat{\theta }_{\delta
},\delta ,l).  \label{b}
\end{equation}%
Putting (\ref{a}) and (\ref{b}) together, we have that%
\begin{equation*}
\widetilde{Q}_{n}(\widehat{t};\widehat{\theta }_{\delta },\delta ,l)=%
\widetilde{Q}_{n}(\widehat{\theta }_{\delta };\widehat{\theta }_{\delta
},\delta ,l)
\end{equation*}%
so that $\widehat{\theta }_{\delta }$ is also a minimizer to the problem (%
\ref{H(t)}) with $t=\widehat{\theta }_{\delta }$. It thus follows that $%
\widehat{\theta }_{\delta }\in H_{\delta ,l}(\widehat{\theta }_{\delta })$.

(b) Proof of part (b) follows closely that of Theorem 3.1 of %
\citet{bertsimas2016}. Note that, for any $l\geq h$, if $t^{\prime }\in
H_{\delta ,l}(t)$, then 
\begin{eqnarray}
Q_{n}(t;\delta ) &=&\widetilde{Q}_{n}(t;t,\delta ,l)  \notag \\
&\geq &\widetilde{Q}_{n}(t^{\prime };t,\delta ,l)  \notag \\
&=&\frac{l-h}{2}\left\Vert t^{\prime }-t\right\Vert _{2}^{2}+\widetilde{Q}%
_{n}(t^{\prime };t,\delta ,h)  \notag \\
&\geq &\frac{l-h}{2}\left\Vert t^{\prime }-t\right\Vert
_{2}^{2}+Q_{n}(t^{\prime };\delta )  \label{c}
\end{eqnarray}%
so that%
\begin{equation}
\left\Vert t^{\prime }-t\right\Vert _{2}^{2}\leq \frac{2\left(
Q_{n}(t;\delta )-Q_{n}(t^{\prime };\delta )\right) }{l-h}.  \label{d}
\end{equation}%
Now let $l>h$ and consider the sequence $t_{m}$ satisfying $t_{m+1}\in
H_{\delta ,l}(t_{m})$. Since the parameter space $\Theta $ is compact,
inequality (\ref{c}) implies that $Q_{n}(t_{m};\delta )$ decreases with $m$
and thus converges to a limit $Q^{\ast }$ as $m\longrightarrow \infty $.
Using this fact together with (\ref{d}), it follows that $t_{m}$ also
converges to a limit $t^{\ast }$. Applying (\ref{d}) with $t^{\prime
}=t_{m+1}$ and $t=t_{m}$, we can deduce (\ref{convergence}) by noting that%
\begin{equation*}
\min_{m=1,...,N}\Vert t_{m+1}-t_{m}\Vert _{2}^{2}\leq \frac{1}{N}%
\sum\nolimits_{m=1}^{N}\Vert t_{m+1}-t_{m}\Vert _{2}^{2}\leq \frac{2\left(
Q_{n}(t_{1};\delta )-Q^{\ast }\right) }{N\left( l-h\right) }.
\end{equation*}
\end{proof}

\subsection{Proof of Proposition \protect\ref{Approximation}}

\begin{proof}[Proof of Proposition \protect\ref{Approximation}]
By (\ref{Sn(theta;delta)}) and (\ref{Qn(theta;delta)}), we have that $%
Q_{n}(\theta ;0)=S_{n}(\theta )+\lambda \Vert \theta \Vert _{0}$.

For each $\theta \in \Theta $, 
\begin{equation*}
Q_{n}(\theta ;\delta )\leq Q_{n}(\theta ;0)\leq Q_{n}(\theta ;\delta )+\frac{%
\delta c_{\tau }}{2}.
\end{equation*}%
Therefore, we can deduce that%
\begin{eqnarray*}
Q_{n}(\widehat{\theta }_{\delta };\delta ) &\leq &\min_{\theta \in \mathbb{B}%
(k_{0})}Q_{n}(\theta ;0) \\
&\leq &Q_{n}(\widehat{\theta }_{\delta };0) \\
&\leq &Q_{n}(\widehat{\theta }_{\delta };\delta )+\frac{\delta c_{\tau }}{2}
\\
&\leq &\min_{\theta \in \mathbb{B}(k_{0})}Q_{n}(\theta ;0)+\frac{\delta
c_{\tau }}{2}.
\end{eqnarray*}
\end{proof}

\section*{\LARGE Online Appendix to ``Sparse Quantile Regression'' by Le-Yu Chen and Sokbae Lee}

\section{Additional Results of the Empirical Application\label{empirical_application_appendix}}

In this online appendix, we provide the variable selection results obtained under $\ell _{0}$-PQR, $\ell _{0}$-CQR, $\ell _{1}$-PQR and the other four
alternative penalized quantile regression approaches (AL-SCAD, AL-MCP,
QR-SCAD, QR-MCP) in our empirical study of Section 6. 
 To facilitate presentation of these results, we give below further details on
the five covariate specifications used in this study.

Table \ref{basic covariate specification} lists names and definitions of the
basic covariates. These variables constitute the first covariate
specification where $p=21$. The second specification builds on and modifies
the first as follows. Let \texttt{medu1}, \texttt{medu2}, \texttt{medu3} be
the binary variables indicating respectively whether mother's years of
education were exactly 12, between 12 and 16, or at least 16. Define
analogously \texttt{fedu1}, \texttt{fedu2}, \texttt{fedu3} to be the
corresponding indicator variables that discretize father's years of
education. For $j=1,...,m+3$, let \texttt{$B_{j}$(mage)}, \texttt{$B_{j}$%
(fage)}, \texttt{$B_{j}$(npre)} and \texttt{$B_{j}$(mslb)} denote the cubic
B-spline series terms for approximating functions of the variables \texttt{%
mage}, \texttt{fage}, \texttt{nprenatal} and \texttt{monthslb} respectively
using $m$ interior knots where these approximations do not include B-spline
intercept terms.

The second specification consists of all variables of the first
specification except that the covariates \texttt{medu} and \texttt{fedu} are
replaced by the six indicator variables that discretize both parents' years
of educations as defined above, and moreover the variables \texttt{mage}, 
\texttt{fage}, \texttt{nprenatal} and \texttt{monthslb} are replaced by
their corresponding cubic B-spline terms using 4 interior knots. The third
covariate specification consists of all variables in the second
specification together with those obtained by interacting the B-spline
expansion terms with the other explanatory variables. Both the fourth and
fifth specifications are constructed using the same procedure as for the
third case except that we increase numbers of interior B-spline knots to be
12 and 16 respectively for these two specifications. Accordingly, the
covariate vector under the second, third, fourth and fifth specifications
has dimension 49, 609, 1281, and 1617 respectively. Finally, for each
covariate specification, all stochastic covariates thus constructed are
further standardized to have mean zero and variance unity.

We now discuss the variable selection results of our empirical study. Tables %
\ref{sel_lower_quantile_21} - \ref{sel_upper_quantile_1617} report results
of the top 10 most often selected covariates as well as their proportions of
being selected and the corresponding average estimated regression
coefficient values. From these tables, we note that the regression intercept
was always selected under every estimation approach and across all the
covariate specifications. In addition, its average estimated value was quite
similar in most estimation scenarios. At 5\% quantile level, for the case
with $p=21$, Table \ref{sel_lower_quantile_21} indicates that the variable
for number of prenatal care visits (\texttt{nprenatal}) was also always
selected and other important predictors were mother's smoking behavior
during pregnancy (\texttt{msmoke}) and her race (\texttt{mrace}), both being
selected with at least 60\% incidence rate across all estimation approaches.
For each of these three variables, the corresponding estimated coefficient
was also of the same sign and had similar magnitude across all the methods.
At 95\% quantile level, analogous variable selection results emerged in
Table \ref{sel_upper_quantile_21} though \texttt{mrace} was no longer listed
among the top 4 most often selected variables under some of the estimation
approaches.

For the covariate specification with $p=49$, Table \ref%
{sel_lower_quantile_49} shows that the B-spline expansion terms \texttt{$%
B_{1}$(npre)}, \texttt{$B_{2}$(npre)} and \texttt{$B_{4}$(npre)} were among
the top 4 most often selected variables across most of the quantile
regression approaches for the estimation at 5\% quantile level. This
indicates that the true 5\% level conditional quantile function could be
nonlinear in the variable \texttt{nprenatal}. Yet, at 95\% quantile level,
we find that, except for \texttt{$B_{5}$(npre)}, which was selected in at
least 60\% of the conformal prediction replications under AL-SCAD and
AL-MCP, covariates of the B-spline series terms appeared to be less
important under the other estimation approaches. By contrast, maternal
smoking behavior (\texttt{msmoke}) was the most often selected stochastic
variable across all the estimation approaches in this setting.

Finally, for higher dimensional cases with $p\in\{609,1281,1617\}$, it is
evident from Tables \ref{sel_lower_quantile_609} - \ref%
{sel_upper_quantile_1617} that, except for the $\ell _{1}$-PQR cases, for
each of the stochastic covariates, its incidence of selection was well
capped below 60\% under all the other estimation approaches. Specifically,
at 95\% quantile level in the case with $p=1617$, regression intercept was
the only selected covariate under both MIO and FO based implementations of $%
\ell _{0}$-PQR. On the whole, while all estimation approaches agreed to the
selection of regression intercept, we note that the variable selection
results generally appeared to vary to a much larger extent across methods
under the high dimensional covariate specifications.

\begin{center}
\begin{table}[tbh]
\caption{Names and definitions of basic covariates }
\label{basic covariate specification}%
\begin{tabular}{l|l}
\hline\hline
Variable name & Definition \\ \hline
intercept & regression intercept \\ 
married & marital status \\ 
mage & mother's age \\ 
medu & mother's years of education \\ 
mhisp & whether mother is hispanic \\ 
mrace & whether mother's race is white \\ 
fage & father's age \\ 
fedu & father's years of education \\ 
fhisp & whether father is hispanic \\ 
frace & whether father's race is white \\ 
foreign & whether mother was born abroad \\ 
alcohol & whether mother drank alcohol during pregnancy \\ 
msmoke & whether mother smoked during pregnancy \\ 
deadkids & whether a newborn died in previous births \\ 
monthslb & number of months since last birth \\ 
nprenatal & number of prenatal care visits \\ 
trimester1 & whether the first prenatal care visit was in the first trimester
\\ 
fbaby & whether the infant was the first born child \\ 
season1 & whether the infant was born in the winter \\ 
season2 & whether the infant was born in the spring \\ 
season3 & whether the infant was born in the summer \\ \hline\hline
\end{tabular}%
\end{table}
\end{center}

\clearpage
\begin{threeparttable}[tbh]
\caption{Top 10 most often selected variables for $p=21$ at $5\%$ quantile}
\label{sel_lower_quantile_21}
\footnotesize{
\begin{tabular}{c|cccccccc}
\hline\hline
& \multicolumn{2}{|c}{$\ell _{0}$-PQR} & $\ell _{0}$-CQR & $\ell _{1}$-PQR & 
AL-SCAD & AL-MCP & QR-SCAD & QR-MCP \\ 
~ & MIO & FO & ~ & ~ & ~ & ~ & ~ & ~ \\ \hline
1st & intercept & intercept & intercept & intercept & intercept & intercept & intercept & intercept \\ 
        ~ & (1,2.47) & (1,2.46) & (1,2.46) & (1,2.44) & (1,2.47) & (1,2.48) & (1,2.47) & (1,2.47) \\ \hline
        2nd & nprenatal & nprenatal & nprenatal & nprenatal & nprenatal & nprenatal & nprenatal & nprenatal \\ 
        ~ & (1,0.16) & (1,0.15) & (1,0.16) & (1,0.14) & (1,0.14) & (1,0.14) & (1,0.17) & (1,0.17) \\ \hline
        3rd & msmoke & msmoke & msmoke & mrace & mrace & mrace & mrace & mrace \\ 
        ~ & (0.7,-0.11) & (0.6,-0.12) & (0.7,-0.1) & (0.9,0.11) & (0.7,0.16) & (0.8,0.16) & (0.9,0.19) & (0.9,0.19) \\ \hline
        4th & mrace & mrace & mrace & msmoke & msmoke & msmoke & msmoke & msmoke \\ 
        ~ & (0.6,0.24) & (0.6,0.18) & (0.7,0.23) & (0.8,-0.06) & (0.6,-0.09) & (0.7,-0.09) & (0.6,-0.11) & (0.6,-0.11) \\ \hline
        5th & trimester1 & medu & trimester1 & foreign & trimester1 & trimester1 & trimester1 & trimester1 \\ 
        ~ & (0.3,-0.12) & (0.3,0.08) & (0.4,-0.09) & (0.6,-0.01) & (0.6,-0.05) & (0.7,-0.07) & (0.6,-0.1) & (0.6,-0.1) \\ \hline
        6th & medu & frace & medu & frace & medu & medu & mhisp & mhisp \\ 
        ~ & (0.2,0.12) & (0.3,0.15) & (0.2,0.12) & (0.6,0.06) & (0.5,0.06) & (0.6,0.06) & (0.3,0.04) & (0.3,0.04) \\ \hline
        7th & frace & trimester1 & frace & married & married & foreign & foreign & foreign \\ 
        ~ & (0.2,0.19) & (0.3,-0.11) & (0.2,0.19) & (0.5,0.03) & (0.3,0.03) & (0.4,-0.02) & (0.3,-0.07) & (0.3,-0.07) \\ \hline
        8th & season2 & season2 & season2 & mage & foreign & married & monthslb & monthslb \\ 
        ~ & (0.2,-0.07) & (0.2,-0.07) & (0.2,-0.07) & (0.5,0.01) & (0.3,-0.02) & (0.3,0.03) & (0.3,-0.02) & (0.3,-0.02) \\ \hline
        9th & foreign & married & married & season1 & monthslb & deadkids & frace & frace \\ 
        ~ & (0.1,-0.08) & (0.1,0.07) & (0.1,0.06) & (0.5,0) & (0.3,0.02) & (0.3,0.01) & (0.3,0.08) & (0.3,0.08) \\ \hline
        10th & mage & foreign & mhisp & season2 & frace & monthslb & fbaby & fbaby \\ 
        ~ & (0.1,-0.05) & (0.1,-0.08) & (0.1,0.11) & (0.5,-0.01) & (0.3,0.14) & (0.3,0.02) & (0.3,-0.02) & (0.3,-0.01) \\ \hline
\end{tabular}
}
\begin{tablenotes}
\item[]
For each parenthesized pair of values, the left value shows the proportion of the variable being selected across the sample splitting replications in the conformalized quantile regression procedure. The right value shows the corresponding averaged estimated regression coefficient value over those cases where the variable has been selected.
\end{tablenotes}
\end{threeparttable}

\begin{threeparttable}[tbh]
\caption{Top 10 most often selected variables for $p=21$ at $95\%$ quantile}
\label{sel_upper_quantile_21}
\footnotesize{
\begin{tabular}{c|cccccccc}
\hline\hline
& \multicolumn{2}{|c}{$\ell _{0}$-PQR} & $\ell _{0}$-CQR & $\ell _{1}$-PQR & 
AL-SCAD & AL-MCP & QR-SCAD & QR-MCP \\ 
~ & MIO & FO & ~ & ~ & ~ & ~ & ~ & ~ \\ \hline
1st & intercept & intercept & intercept & intercept & intercept & intercept & intercept & intercept \\ 
        ~ & (1,4.2) & (1,4.19) & (1,4.2) & (1,4.16) & (1,4.19) & (1,4.2) & (1,4.19) & (1,4.2) \\ \hline
        2nd & msmoke & msmoke & msmoke & nprenatal & msmoke & msmoke & msmoke & msmoke \\ 
        ~ & (1,-0.11) & (1,-0.1) & (1,-0.11) & (1,0.06) & (1,-0.1) & (1,-0.11) & (1,-0.11) & (1,-0.11) \\ \hline
        3rd & nprenatal & nprenatal & nprenatal & msmoke & nprenatal & mrace & nprenatal & nprenatal \\ 
        ~ & (0.5,0.09) & (0.7,0.08) & (0.7,0.08) & (1,-0.08) & (0.7,0.09) & (0.8,0.05) & (0.4,0.11) & (0.5,0.09) \\ \hline
        4th & mrace & alcohol & mrace & mrace & season1 & nprenatal & mrace & mrace \\ 
        ~ & (0.5,0.08) & (0.5,-0.04) & (0.6,0.07) & (0.9,0.03) & (0.7,0.06) & (0.7,0.08) & (0.4,0.07) & (0.5,0.08) \\ \hline
        5th & alcohol & fage & alcohol & alcohol & mrace & frace & fage & alcohol \\ 
        ~ & (0.4,-0.02) & (0.4,0.02) & (0.5,-0.03) & (0.8,-0.03) & (0.6,0.07) & (0.5,-0.01) & (0.3,-0.02) & (0.4,-0.02) \\ \hline
        6th & season1 & mrace & fbaby & medu & fbaby & fbaby & trimester1 & season1 \\ 
        ~ & (0.4,0.07) & (0.4,0.08) & (0.5,-0.06) & (0.8,0.02) & (0.5,-0.05) & (0.5,-0.05) & (0.3,-0.01) & (0.4,0.07) \\ \hline
        7th & married & fbaby & fage & foreign & trimester1 & season1 & season1 & fhisp \\ 
        ~ & (0.3,0) & (0.4,-0.08) & (0.4,0.02) & (0.7,-0.01) & (0.5,-0.02) & (0.5,0.08) & (0.3,0.08) & (0.3,-0.02) \\ \hline
        8th & medu & trimester1 & season1 & fage & married & married & married & mage \\ 
        ~ & (0.3,0.05) & (0.4,-0.03) & (0.4,0.07) & (0.7,0.01) & (0.4,-0.02) & (0.4,-0.01) & (0.2,0.03) & (0.3,0.03) \\ \hline
        9th & fage & season1 & married & fbaby & fhisp & fage & fhisp & frace \\ 
        ~ & (0.3,0.01) & (0.4,0.08) & (0.3,-0.01) & (0.7,-0.04) & (0.4,-0.03) & (0.4,0.01) & (0.2,-0.03) & (0.3,-0.06) \\ \hline
        10th & frace & foreign & trimester1 & season1 & frace & trimester1 & alcohol & trimester1 \\ 
        ~ & (0.3,-0.05) & (0.3,0) & (0.3,0) & (0.7,0.04) & (0.4,-0.01) & (0.4,-0.02) & (0.2,-0.02) & (0.3,0.01) \\ \hline
\end{tabular}
}
\begin{tablenotes} \item[]
For each parenthesized pair of values, the left value shows the proportion of the variable being selected across the sample splitting replications in the conformalized quantile regression procedure. The right value shows the corresponding averaged estimated regression coefficient value over those cases where the variable has been selected.
\end{tablenotes}
\end{threeparttable}

\begin{threeparttable}[tbh]
\caption{Top 10 most often selected variables for $p=49$ at $5\%$ quantile}
\label{sel_lower_quantile_49}
\footnotesize{
\begin{tabular}{c|cccccccc}
\hline\hline
& \multicolumn{2}{|c}{$\ell _{0}$-PQR} & $\ell _{0}$-CQR & $\ell _{1}$-PQR & 
AL-SCAD & AL-MCP & QR-SCAD & QR-MCP \\ 
~ & MIO & FO & ~ & ~ & ~ & ~ & ~ & ~ \\ \hline
1st & intercept & intercept & intercept & intercept & intercept & intercept & intercept & intercept \\ 
        ~ & (1,2.47) & (1,2.45) & (1,2.47) & (1,2.44) & (1,2.46) & (1,2.47) & (1,2.45) & (1,2.44) \\ \hline
        2nd & B4(npre) & B1(npre) & B4(npre) & B4(npre) & B4(npre) & B4(npre) & B4(npre) & B4(npre) \\ 
        ~ & (1,0.3) & (0.8,-0.26) & (1,0.38) & (1,0.11) & (1,0.25) & (1,0.22) & (0.9,0.27) & (0.8,0.27) \\ \hline
        3rd & frace & B4(npre) & frace & B1(npre) & B1(npre) & B1(npre) & B1(npre) & B1(npre) \\ 
        ~ & (0.6,0.12) & (0.8,0.21) & (0.7,0.11) & (0.9,-0.14) & (0.8,-0.26) & (0.8,-0.19) & (0.7,-0.33) & (0.7,-0.32) \\ \hline
        4th & B2(npre) & frace & trimester1 & mrace & B2(npre) & B2(npre) & B2(npre) & B2(npre) \\ 
        ~ & (0.6,0.23) & (0.5,0.1) & (0.6,-0.11) & (0.7,0.06) & (0.6,0.23) & (0.5,0.17) & (0.6,0.22) & (0.5,0.23) \\ \hline
        5th & trimester1 & msmoke & B2(npre) & B3(mslb) & mrace & mrace & B4(mslb) & mrace \\ 
        ~ & (0.5,-0.11) & (0.4,-0.11) & (0.6,0.25) & (0.7,-0.03) & (0.5,0.12) & (0.4,0.13) & (0.4,0.18) & (0.3,0.1) \\ \hline
        6th & msmoke & B1(fage) & B3(npre) & frace & fbaby & trimester1 & msmoke & B1(fage) \\ 
        ~ & (0.4,-0.12) & (0.4,-0.11) & (0.6,0.31) & (0.6,0.08) & (0.4,0.06) & (0.4,-0.06) & (0.3,-0.12) & (0.3,-0.09) \\ \hline
        7th & mrace & mrace & B1(npre) & B1(fage) & B1(fage) & B1(fage) & mrace & msmoke \\ 
        ~ & (0.4,0.14) & (0.3,0.13) & (0.4,-0.3) & (0.6,-0.06) & (0.4,-0.04) & (0.4,-0.05) & (0.3,0.16) & (0.2,-0.11) \\ \hline
        8th & B1(npre) & B3(npre) & msmoke & msmoke & B3(mslb) & B3(mslb) & fbaby & frace \\ 
        ~ & (0.4,-0.38) & (0.3,0.11) & (0.3,-0.1) & (0.5,-0.05) & (0.4,-0.07) & (0.4,-0.07) & (0.3,0.09) & (0.2,0.11) \\ \hline
        9th & B3(npre) & medu1 & B5(npre) & trimester1 & B4(mslb) & B4(mslb) & trimester1 & fbaby \\ 
        ~ & (0.4,0.28) & (0.2,-0.09) & (0.3,0.18) & (0.5,-0.03) & (0.4,0.15) & (0.4,0.13) & (0.3,-0.13) & (0.2,0.09) \\ \hline
        10th & mhisp & B4(mage) & mrace & B4(mslb) & msmoke & msmoke & B3(mslb) & trimester1 \\ 
        ~ & (0.2,0.09) & (0.2,0) & (0.2,0.15) & (0.5,0.03) & (0.3,-0.1) & (0.3,-0.09) & (0.3,-0.11) & (0.2,-0.1) \\ \hline
\end{tabular}
}
\begin{tablenotes} \item[]
For each parenthesized pair of values, the left value shows the proportion of the variable being selected across the sample splitting replications in the conformalized quantile regression procedure. The right value shows the corresponding averaged estimated regression coefficient value over those cases where the variable has been selected.
\end{tablenotes}
\end{threeparttable}

\begin{threeparttable}[tbh]
\caption{Top 10 most often selected variables for $p=49$ at $95\%$ quantile}
\label{sel_upper_quantile_49}
\footnotesize{
\begin{tabular}{c|cccccccc}
\hline\hline
& \multicolumn{2}{|c}{$\ell _{0}$-PQR} & $\ell _{0}$-CQR & $\ell _{1}$-PQR & 
AL-SCAD & AL-MCP & QR-SCAD & QR-MCP \\ 
~ & MIO & FO & ~ & ~ & ~ & ~ & ~ & ~ \\ \hline
1st & intercept & intercept & intercept & intercept & intercept & intercept & intercept & intercept \\ 
        ~ & (1,4.19) & (1,4.19) & (1,4.19) & (1,4.15) & (1,4.19) & (1,4.19) & (1,4.19) & (1,4.2) \\ \hline
        2nd & msmoke & msmoke & msmoke & msmoke & msmoke & msmoke & msmoke & msmoke \\ 
        ~ & (0.9,-0.1) & (1,-0.1) & (1,-0.11) & (1,-0.07) & (0.8,-0.1) & (0.9,-0.09) & (0.8,-0.11) & (0.7,-0.11) \\ \hline
        3rd & mrace & mrace & mrace & alcohol & B5(npre) & B5(npre) & B5(npre) & B5(npre) \\ 
        ~ & (0.4,0.07) & (0.3,0.06) & (0.3,0.08) & (0.8,-0.02) & (0.6,0.08) & (0.8,0.07) & (0.5,0.1) & (0.4,0.12) \\ \hline
        4th & B5(npre) & B1(npre) & fedu2 & mrace & season1 & fedu2 & B6(npre) & B6(mslb) \\ 
        ~ & (0.3,0.12) & (0.3,-0.06) & (0.2,0.06) & (0.8,0.03) & (0.4,0.04) & (0.5,0.04) & (0.4,-0.08) & (0.4,-0.09) \\ \hline
        5th & fbaby & season1 & B1(fage) & fedu2 & B1(fage) & mrace & fedu1 & B3(npre) \\ 
        ~ & (0.2,-0.06) & (0.2,0.07) & (0.2,-0.08) & (0.7,0.02) & (0.4,-0.06) & (0.4,0.03) & (0.3,0.08) & (0.3,0.1) \\ \hline
        6th & B6(mage) & B4(mage) & B1(npre) & B1(fage) & B5(fage) & season1 & fedu2 & B6(npre) \\ 
        ~ & (0.2,0.09) & (0.2,-0.07) & (0.2,-0.07) & (0.7,-0.02) & (0.4,0.02) & (0.4,0.03) & (0.3,0.07) & (0.3,-0.08) \\ \hline
        7th & B1(fage) & B1(fage) & B4(npre) & B1(npre) & married & fedu1 & B6(mage) & B7(mslb) \\ 
        ~ & (0.2,-0.1) & (0.2,-0.06) & (0.2,0.09) & (0.7,-0.04) & (0.3,-0.01) & (0.4,0.03) & (0.3,0.09) & (0.3,0.12) \\ \hline
        8th & B3(fage) & B5(fage) & frace & B5(npre) & season2 & B1(fage) & B3(npre) & married \\ 
        ~ & (0.2,0.01) & (0.2,0.01) & (0.1,0.08) & (0.7,0.02) & (0.3,0.02) & (0.4,-0.07) & (0.3,0.11) & (0.2,-0.02) \\ \hline
        9th & B1(npre) & married & season1 & B4(mslb) & fedu2 & B5(fage) & B4(npre) & mrace \\ 
        ~ & (0.2,-0.02) & (0.1,0.09) & (0.1,0.09) & (0.7,0.02) & (0.3,0.05) & (0.4,0.01) & (0.3,0.14) & (0.2,0.09) \\ \hline
        10th & B6(npre) & fbaby & medu2 & season1 & B3(mage) & B6(npre) & B6(mslb) & fedu1 \\ 
        ~ & (0.2,-0.07) & (0.1,-0.04) & (0.1,0.06) & (0.5,0.02) & (0.3,0.01) & (0.4,-0.07) & (0.3,-0.08) & (0.2,0.11) \\ \hline
\end{tabular}
}
\begin{tablenotes} \item[]
For each parenthesized pair of values, the left value shows the proportion of the variable being selected across the sample splitting replications in the conformalized quantile regression procedure. The right value shows the corresponding averaged estimated regression coefficient value over those cases where the variable has been selected.
\end{tablenotes}
\end{threeparttable}

\begin{landscape}
\begin{threeparttable}[tbh]
\caption{Top 10 most often selected variables for $p=609$ at $5\%$ quantile}
\label{sel_lower_quantile_609}
\scriptsize{
\begin{tabular}{c|cccccccc}
\hline\hline
& \multicolumn{2}{|c}{$\ell _{0}$-PQR} & $\ell _{0}$-CQR & $\ell _{1}$-PQR & 
AL-SCAD & AL-MCP & QR-SCAD & QR-MCP \\ 
~ & MIO & FO & ~ & ~ & ~ & ~ & ~ & ~ \\ \hline
1st & intercept & intercept & intercept & intercept & intercept & intercept & intercept & intercept \\ 
        ~ & (1,2.44) & (1,2.46) & (1,2.45) & (1,2.42) & (1,2.45) & (1,2.45) & (1,2.44) & (1,2.44) \\ \hline
        2nd & B1(npre) & trimester1*B1(npre) & trimester1*B1(npre) & trimester1*B1(npre) & foreign*B7(npre) & mrace & trimester1*B1(npre) & trimester1*B1(npre) \\ 
        ~ & (0.2,-0.45) & (0.4,-0.23) & (0.5,-0.26) & (0.8,-0.09) & (0.3,-0.22) & (0.2,0.05) & (0.2,-0.19) & (0.2,-0.26) \\ \hline
        3rd & B4(npre) & frace*B4(npre) & frace*B4(npre) & mrace & mrace & foreign*B7(npre) & season2 & B4(npre) \\ 
        ~ & (0.2,0.23) & (0.3,0.2) & (0.4,0.21) & (0.6,0.04) & (0.2,0.03) & (0.2,-0.34) & (0.1,-0.4) & (0.1,0.23) \\ \hline
        4th & married*B2(npre) & foreign*B7(npre) & B4(npre) & B4(npre) & trimester1*B1(npre) & trimester1*B1(npre) & B1(fage) & foreign*B7(npre) \\ 
        ~ & (0.1,-0.01) & (0.2,-0.75) & (0.3,0.17) & (0.6,0.06) & (0.2,-0.18) & (0.2,-0.21) & (0.1,-0.15) & (0.1,-0.69) \\ \hline
        5th & frace*B4(npre) & mrace & B1(fage) & frace*B4(npre) & B4(npre) & B1(npre) & B1(npre) & alcohol*B7(npre) \\ 
        ~ & (0.1,0.29) & (0.1,0.14) & (0.1,-0.16) & (0.6,0.06) & (0.1,0.13) & (0.1,-0.01) & (0.1,-0.47) & (0.1,0) \\ \hline
        6th & frace*B5(mslb) & frace & married*B2(npre) & frace & mhisp*B7(fage) & B4(npre) & B4(npre) & mrace*B2(npre) \\ 
        ~ & (0.1,0.04) & (0.1,0.16) & (0.1,0.15) & (0.4,0.05) & (0.1,0.01) & (0.1,0.13) & (0.1,0.23) & (0.1,0.22) \\ \hline
        7th & fedu1*B1(npre) & B1(fage) & married*B4(npre) & mrace*B4(npre) & deadkids*B1(npre) & married*B3(mslb) & B4(mslb) & mrace*B4(npre) \\ 
        ~ & (0.1,-0.35) & (0.1,-0.11) & (0.1,0.08) & (0.4,0.02) & (0.1,-0.05) & (0.1,-0.01) & (0.1,0.13) & (0.1,0.27) \\ \hline
        8th & ~ & B4(npre) & foreign*B6(npre) & fedu1*B1(npre) & msmoke*B4(mage) & mhisp*B7(fage) & married*B3(mslb) & fedu1*B1(npre) \\ 
        ~ & ~ & (0.1,0.12) & (0.1,0) & (0.4,-0.07) & (0.1,-0.09) & (0.1,0.01) & (0.1,-0.16) & (0.1,-0.35) \\ \hline
        9th & ~ & deadkids*B1(npre) & alcohol*B5(npre) & B1(fage) & mrace*B3(mage) & alcohol*B7(mslb) & fhisp*B7(mslb) & fedu1*B4(npre) \\ 
        ~ & ~ & (0.1,-0.1) & (0.1,0) & (0.3,-0.03) & (0.1,0.11) & (0.1,-0.15) & (0.1,0.31) & (0.1,0.19) \\ \hline
        10th & ~ & msmoke*B3(mage) & mrace*B4(npre) & msmoke*B4(mage) & mrace*B4(npre) & deadkids*B1(npre) & foreign*B7(npre) & fedu2*B1(npre) \\ 
        ~ & ~ & (0.1,-0.15) & (0.1,0.25) & (0.3,-0.02) & (0.1,0.13) & (0.1,-0.06) & (0.1,-0.3) & (0.1,-0.18) \\ \hline
\end{tabular}
}
\begin{tablenotes} \item[]
For each parenthesized pair of values, the left value shows the proportion of the variable being selected across the sample splitting replications in the conformalized quantile regression procedure. The right value shows the corresponding averaged estimated regression coefficient value over those cases where the variable has been selected.
\end{tablenotes}
\end{threeparttable}
\newline
\vspace{4em}
\newline

\begin{threeparttable}[tbh]       
\caption{Top 10 most often selected variables for $p=609$ at $95\%$ quantile}
\label{sel_upper_quantile_609}
\scriptsize{
\begin{tabular}{c|cccccccc}
\hline\hline
& \multicolumn{2}{|c}{$\ell _{0}$-PQR} & $\ell _{0}$-CQR & $\ell _{1}$-PQR & 
AL-SCAD & AL-MCP & QR-SCAD & QR-MCP \\ 
~ & MIO & FO & ~ & ~ & ~ & ~ & ~ & ~ \\ \hline
1st & intercept & intercept & intercept & intercept & intercept & intercept & intercept & intercept \\ 
        ~ & (1,4.22) & (1,4.21) & (1,4.2) & (1,4.03) & (1,4.17) & (1,4.17) & (1,4.18) & (1,4.16) \\ \hline
        2nd & msmoke & msmoke & msmoke & msmoke & mhisp*B7(mage) & mhisp*B7(mage) & married*B5(npre) & mhisp*B7(mage) \\ 
        ~ & (0.1,-0.08) & (0.1,-0.06) & (0.3,-0.12) & (0.6,-0.05) & (0.2,-1.55) & (0.2,-2.1) & (0.2,0.08) & (0.2,-9.11) \\ \hline
        3rd & B1(npre) & B4(mage) & married*B4(npre) & mhisp*B7(npre) & season1*B7(npre) & season1*B7(npre) & mhisp*B7(mage) & married*B5(npre) \\ 
        ~ & (0.1,-0.06) & (0.1,-0.16) & (0.2,0.06) & (0.4,0) & (0.2,-0.2) & (0.2,-0.21) & (0.2,0.89) & (0.1,0.01) \\ \hline
        4th & deadkids*B2(fage) & B6(mage) & msmoke*B3(fage) & mrace & trimester1 & trimester1 & season1*B7(npre) & frace*B4(npre) \\ 
        ~ & (0.1,-0.03) & (0.1,0.04) & (0.2,-0.15) & (0.3,0.01) & (0.1,0.09) & (0.1,0.07) & (0.2,-0.29) & (0.1,0.01) \\ \hline
        5th & season1*B5(npre) & B4(npre) & mrace*B4(mslb) & married*B4(npre) & B3(mage) & B3(mage) & fbaby & season1*B3(npre) \\ 
        ~ & (0.1,0.01) & (0.1,0.2) & (0.2,0.09) & (0.3,-0.01) & (0.1,0.12) & (0.1,0.14) & (0.1,-0.07) & (0.1,0.09) \\ \hline
        6th & season1*B7(mslb) & married*B3(fage) & frace*B4(fage) & married*B5(npre) & B4(mage) & B4(mage) & B1(mage) & season1*B7(npre) \\ 
        ~ & (0.1,-0.06) & (0.1,0.07) & (0.2,0.03) & (0.3,0.03) & (0.1,-0.11) & (0.1,-0.1) & (0.1,0.19) & (0.1,-0.38) \\ \hline
        7th & ~ & married*B5(fage) & frace*B4(npre) & mrace*B4(mslb) & B4(fage) & B4(fage) & B3(mage) & medu2*B3(mslb) \\ 
        ~ & ~ & (0.1,-0.12) & (0.2,0.1) & (0.3,0.01) & (0.1,-0.01) & (0.1,-0.03) & (0.1,0.15) & (0.1,0) \\ \hline
        8th & ~ & married*B1(npre) & fbaby*B4(npre) & season2*B3(mslb) & B5(fage) & B5(fage) & B2(fage) & ~ \\ 
        ~ & ~ & (0.1,-0.07) & (0.2,0.06) & (0.3,0.02) & (0.1,-0.04) & (0.1,0) & (0.1,-0.14) & ~ \\ \hline
        9th & ~ & married*B4(npre) & season1*B5(fage) & trimester1 & B6(fage) & B6(fage) & B3(fage) & ~ \\ 
        ~ & ~ & (0.1,-0.11) & (0.2,0.07) & (0.2,0) & (0.1,-0.04) & (0.1,-0.06) & (0.1,-0.09) & ~ \\ \hline
        10th & ~ & married*B5(npre) & B6(fage) & B1(npre) & B1(npre) & B1(npre) & B3(npre) & ~ \\ 
        ~ & ~ & (0.1,0.07) & (0.1,0.02) & (0.2,-0.04) & (0.1,0.07) & (0.1,0.07) & (0.1,0.08) \\ \hline
\end{tabular}
}
\begin{tablenotes} \item[]
For each parenthesized pair of values, the left value shows the proportion of the variable being selected across the sample splitting replications in the conformalized quantile regression procedure. The right value shows the corresponding averaged estimated regression coefficient value over those cases where the variable has been selected.
\end{tablenotes}
\end{threeparttable}

\newpage
\begin{threeparttable}[tbh]
\caption{Top 10 most often selected variables for $p=1281$ at $5\%$ quantile}
\label{sel_lower_quantile_1281}
\scriptsize{
\begin{tabular}{c|cccccccc}
\hline\hline
& \multicolumn{2}{|c}{$\ell _{0}$-PQR} & $\ell _{0}$-CQR & $\ell _{1}$-PQR & 
AL-SCAD & AL-MCP & QR-SCAD & QR-MCP \\ 
~ & MIO & FO & ~ & ~ & ~ & ~ & ~ & ~ \\ \hline
1st & intercept & intercept & intercept & intercept & intercept & intercept & intercept & intercept \\ 
        ~ & (1,2.44) & (1,2.44) & (1,2.44) & (1,2.03) & (1,2.03) & (1,2.04) & (1,2.03) & (1,2.03) \\ \hline
        2nd & trimester1*B2(npre) & trimester1*B2(npre) & trimester1*B1(npre) & mrace & fedu3*B1(fage) & fedu3*B1(fage) & fedu3*B1(fage) & fedu3*B1(fage) \\ 
        ~ & (0.3,-0.26) & (0.3,-0.26) & (0.3,-0.32) & (0.7,0.05) & (0.5,-10) & (0.5,-10) & (0.5,-10) & (0.5,-10) \\ \hline
        3rd & mrace & mrace & B1(npre) & frace & mrace & trimester1*B2(npre) & fedu3*B2(fage) & fedu3*B2(fage) \\ 
        ~ & (0.2,0.13) & (0.2,0.13) & (0.2,-0.26) & (0.7,0.04) & (0.3,0.12) & (0.4,-0.14) & (0.3,-10) & (0.3,-10) \\ \hline
        4th & trimester1*B1(npre) & trimester1*B1(npre) & B12(fage) & trimester1*B2(npre) & trimester1*B2(npre) & fedu1*B1(npre) & mrace & fhisp*B15(mage) \\ 
        ~ & (0.2,-0.28) & (0.2,-0.28) & (0.1,-0.35) & (0.7,-0.08) & (0.3,-0.16) & (0.4,-0.14) & (0.2,0.14) & (0.2,-10) \\ \hline
        5th & medu2*B15(mage) & fedu1*B1(npre) & B2(npre) & fedu3*B1(fage) & fedu3*B2(fage) & mrace & fhisp*B15(mage) & mrace \\ 
        ~ & (0.2,-0.15) & (0.2,-0.23) & (0.1,-0.24) & (0.5,-10) & (0.3,-10) & (0.3,0.12) & (0.2,-10) & (0.1,0.16) \\ \hline
        6th & frace & medu2*B15(mage) & B13(mslb) & msmoke*B7(mage) & married*B1(npre) & fedu3*B2(fage) & fhisp*B15(fage) & married*B1(npre) \\ 
        ~ & (0.1,0.11) & (0.2,-0.15) & (0.1,-0.12) & (0.4,-0.01) & (0.2,-0.1) & (0.3,-10) & (0.2,-5) & (0.1,-0.28) \\ \hline
        7th & foreign*B1(fage) & frace & married*B9(mage) & mrace*B7(mslb) & fhisp*B15(mage) & married*B1(npre) & trimester1*B2(npre) & mhisp*B14(mage) \\ 
        ~ & (0.1,-1.3) & (0.1,0.11) & (0.1,0.07) & (0.4,-0.04) & (0.2,-10) & (0.2,-0.13) & (0.2,-0.27) & (0.1,5.41) \\ \hline
        8th & foreign*B15(npre) & foreign*B1(fage) & mhisp*B6(npre) & trimester1*B1(npre) & season2*B6(fage) & fhisp*B15(mage) & msmoke & mhisp*B15(mage) \\ 
        ~ & (0.1,-1.31) & (0.1,-1.3) & (0.1,0.06) & (0.4,-0.09) & (0.2,-0.07) & (0.2,-10) & (0.1,-0.16) & (0.1,-10) \\ \hline
        9th & alcohol*B6(mage) & foreign*B15(npre) & fhisp*B4(fage) & fedu1*B1(npre) & fedu1*B1(npre) & msmoke & B9(mage) & mhisp*B10(npre) \\ 
        ~ & (0.1,-0.13) & (0.1,-1.31) & (0.1,-0.04) & (0.4,-0.11) & (0.2,-0.17) & (0.1,-0.01) & (0.1,-0.35) & (0.1,-10) \\ \hline
        10th & deadkids*B2(npre) & alcohol*B6(mage) & alcohol*B7(mage) & married & B2(fage) & B2(fage) & married*B9(mage) & mhisp*B15(npre) \\ 
        ~ & (0.1,-0.21) & (0.1,-0.13) & (0.1,-0.02) & (0.3,0.01) & (0.1,-0.05) & (0.1,-0.07) & (0.1,0.32) & (0.1,-10) \\ \hline
\end{tabular}
}
\begin{tablenotes} \item[]
For each parenthesized pair of values, the left value shows the proportion of the variable being selected across the sample splitting replications in the conformalized quantile regression procedure. The right value shows the corresponding averaged estimated regression coefficient value over those cases where the variable has been selected.
\end{tablenotes}
\end{threeparttable}

\begin{threeparttable}[tbh]
\caption{Top 10 most often selected variables for $p=1281$ at $95\%$ quantile}
\label{sel_upper_quantile_1281}
\scriptsize{
\begin{tabular}{c|cccccccc}
\hline\hline
& \multicolumn{2}{|c}{$\ell _{0}$-PQR} & $\ell _{0}$-CQR & $\ell _{1}$-PQR & 
AL-SCAD & AL-MCP & QR-SCAD & QR-MCP \\ 
~ & MIO & FO & ~ & ~ & ~ & ~ & ~ & ~ \\ \hline
1st & intercept & intercept & intercept & intercept & intercept & intercept & intercept & intercept \\ 
        ~ & (1,4.22) & (1,4.21) & (1,4.22) & (1,3.56) & (1,3.75) & (1,3.66) & (1,3.76) & (1,3.76) \\ \hline
        2nd & msmoke & msmoke & msmoke & msmoke & fedu3*B1(fage) & fedu3*B1(fage) & fedu3*B1(fage) & fedu3*B1(fage) \\ 
        ~ & (0.1,-0.16) & (0.1,-0.16) & (0.1,-0.15) & (0.6,-0.06) & (0.5,-10) & (0.5,-10) & (0.5,-10) & (0.5,-10) \\ \hline
        3rd & season1*B12(mslb) & mhisp*B14(mage) & B8(mage) & fedu3*B1(fage) & fedu3*B2(fage) & fedu3*B2(fage) & fedu3*B2(fage) & fedu3*B2(fage) \\ 
        ~ & (0.1,0.06) & (0.1,-3.81) & (0.1,-0.03) & (0.5,-10) & (0.3,-10) & (0.3,-10) & (0.3,-10) & (0.3,-10) \\ \hline
        4th & ~ & ~ & B2(fage) & married*B12(npre) & fhisp*B15(mage) & fhisp*B15(mage) & mhisp*B15(mslb) & fhisp*B15(mage) \\ 
        ~ & ~ & ~ & (0.1,-0.05) & (0.4,0.03) & (0.2,-10) & (0.2,-10) & (0.2,-5.73) & (0.2,-10) \\ \hline
        5th & ~ & ~ & alcohol*B8(mage) & fedu2*B13(mage) & mhisp*B11(mage) & medu2*B10(npre) & fhisp*B15(mage) & mhisp*B11(mage) \\ 
        ~ & ~ & ~ & (0.1,0) & (0.4,0) & (0.1,-0.51) & (0.2,0.02) & (0.2,-10) & (0.1,-0.51) \\ \hline
        6th & ~ & ~ & alcohol*B11(mslb) & medu2*B9(mslb) & mhisp*B14(mage) & mrace & fhisp*B15(mslb) & mhisp*B14(mage) \\ 
        ~ & ~ & ~ & (0.1,-0.01) & (0.4,0.02) & (0.1,-3.81) & (0.1,0.11) & (0.2,1.37) & (0.1,-3.81) \\ \hline
        7th & ~ & ~ & deadkids*B4(mage) & mrace & mhisp*B15(mage) & B9(npre) & mrace & mhisp*B15(mage) \\ 
        ~ & ~ & ~ & (0.1,-0.01) & (0.3,0.03) & (0.1,-10) & (0.1,0.03) & (0.1,0.12) & (0.1,-10) \\ \hline
        8th & ~ & ~ & deadkids*B13(npre) & B1(npre) & mhisp*B10(npre) & B11(mslb) & married*B10(mage) & mhisp*B10(npre) \\ 
        ~ & ~ & ~ & (0.1,-0.03) & (0.3,-0.03) & (0.1,-10) & (0.1,0.09) & (0.1,-0.11) & (0.1,-10) \\ \hline
        9th & ~ & ~ & mrace*B10(mslb) & mhisp*B15(npre) & mhisp*B15(npre) & married*B1(mage) & married*B13(fage) & mhisp*B15(npre) \\ 
        ~ & ~ & ~ & (0.1,0.03) & (0.3,-3.33) & (0.1,-10) & (0.1,-0.01) & (0.1,0.08) & (0.1,-10) \\ \hline
        10th & ~ & ~ & frace*B15(fage) & deadkids*B12(mslb) & mhisp*B15(mslb) & married*B3(mage) & mhisp*B11(mage) & mhisp*B15(mslb) \\ 
        ~ & ~ & ~ & (0.1,-0.05) & (0.3,-0.01) & (0.1,-10) & (0.1,0.04) & (0.1,-0.51) & (0.1,-10) \\ \hline
\end{tabular}
}
\begin{tablenotes} \item[]
For each parenthesized pair of values, the left value shows the proportion of the variable being selected across the sample splitting replications in the conformalized quantile regression procedure. The right value shows the corresponding averaged estimated regression coefficient value over those cases where the variable has been selected.
\end{tablenotes}
\end{threeparttable}
\newpage
\begin{threeparttable}[tbh]
\caption{Top 10 most often selected variables for $p=1617$ at $5\%$ quantile}
\label{sel_lower_quantile_1617}
\scriptsize{
\begin{tabular}{c|cccccccc}
\hline\hline
& \multicolumn{2}{|c}{$\ell _{0}$-PQR} & $\ell _{0}$-CQR & $\ell _{1}$-PQR & 
AL-SCAD & AL-MCP & QR-SCAD & QR-MCP \\ 
~ & MIO & FO & ~ & ~ & ~ & ~ & ~ & ~ \\ \hline
 1st & intercept & intercept & intercept & intercept & intercept & intercept & intercept & intercept \\ 
        ~ & (1,2.46) & (1,2.46) & (1,2.44) & (0.9,1.76) & (0.9,1.79) & (0.9,1.79) & (0.9,1.79) & (0.9,1.79) \\ \hline
        2nd & frace & frace & msmoke*B5(mage) & frace & trimester1*B2(npre) & trimester1*B2(npre) & trimester1*B2(npre) & trimester1*B2(npre) \\ 
        ~ & (0.2,0.11) & (0.2,0.11) & (0.2,-0.04) & (0.6,0.04) & (0.3,-0.1) & (0.4,-0.13) & (0.3,-0.28) & (0.3,-0.27) \\ \hline
        3rd & trimester1*B2(npre) & trimester1*B2(npre) & B11(mage) & trimester1*B2(npre) & fedu3*B2(fage) & mrace & fedu3*B2(fage) & fedu3*B2(fage) \\ 
        ~ & (0.2,-0.28) & (0.2,-0.28) & (0.1,0.07) & (0.6,-0.09) & (0.3,-10) & (0.3,0.1) & (0.3,-10) & (0.3,-10) \\ \hline
        4th & fedu1*B1(npre) & fedu1*B1(npre) & B12(mage) & mrace & medu3*B2(mage) & mrace*B9(mon) & medu3*B2(mage) & medu3*B2(mage) \\ 
        ~ & (0.2,-0.14) & (0.2,-0.14) & (0.1,-0.02) & (0.5,0.04) & (0.3,-10) & (0.3,-0.03) & (0.3,-10) & (0.3,-10) \\ \hline
        5th & fedu1*B2(npre) & fedu1*B2(npre) & B15(mage) & B15(npre) & mrace & fedu3*B2(fage) & mhisp*B19(mage) & mhisp*B19(mage) \\ 
        ~ & (0.2,-0.24) & (0.2,-0.24) & (0.1,-0.12) & (0.5,0.02) & (0.2,0.12) & (0.3,-10) & (0.2,-10) & (0.2,-10) \\ \hline
        6th & fhisp & fhisp & married*B11(mage) & mrace*B9(mslb) & mhisp*B19(mage) & medu3*B2(mage) & mhisp*B13(npre) & mhisp*B13(npre) \\ 
        ~ & (0.1,0.05) & (0.1,0.05) & (0.1,0.18) & (0.4,-0.04) & (0.2,-10) & (0.3,-10) & (0.2,-10) & (0.2,-10) \\ \hline
        7th & fhisp*B14(mage) & fhisp*B14(mage) & married*B12(mage) & trimester1*B1(npre) & mhisp*B13(npre) & frace & fhisp*B14(mage) & fhisp*B14(mage) \\ 
        ~ & (0.1,-0.08) & (0.1,-0.08) & (0.1,0.3) & (0.4,-0.06) & (0.2,-10) & (0.2,0.03) & (0.2,-10) & (0.2,-10) \\ \hline
        8th & foreign*B17(mage) & foreign*B17(mage) & married*B19(fage) & fedu1*B2(npre) & fhisp*B14(mage) & mhisp*B19(mage) & foreign*B19(npre) & fedu1*B1(npre) \\ 
        ~ & (0.1,-0.08) & (0.1,-0.08) & (0.1,-0.1) & (0.4,-0.08) & (0.2,-10) & (0.2,-10) & (0.2,-0.68) & (0.2,-0.25) \\ \hline
        9th & foreign*B2(fage) & foreign*B2(fage) & mhisp*B4(mage) & fedu3*B5(mage) & frace & mhisp*B13(npre) & fedu1*B1(npre) & mrace \\ 
        ~ & (0.1,-0.24) & (0.1,-0.24) & (0.1,0) & (0.4,-0.01) & (0.1,0.02) & (0.2,-10) & (0.2,-0.25) & (0.1,0.22) \\ \hline
        10th & foreign*B19(npre) & foreign*B19(npre) & mhisp*B8(mage) & fhisp*B14(mage) & B2(fage) & fhisp*B14(mage) & mrace & B2(fage) \\ 
        ~ & (0.1,-1.25) & (0.1,-1.25) & (0.1,0.02) & (0.3,-6.69) & (0.1,-0.09) & (0.2,-10) & (0.1,0.22) & (0.1,-0.05) \\ \hline
\end{tabular}
}
\begin{tablenotes} \item[]
For each parenthesized pair of values, the left value shows the proportion of the variable being selected across the sample splitting replications in the conformalized quantile regression procedure. The right value shows the corresponding averaged estimated regression coefficient value over those cases where the variable has been selected.
\end{tablenotes}
\end{threeparttable}

\begin{threeparttable}[tbh]
\caption{Top 10 most often selected variables for $p=1617$ at $95\%$ quantile}
\label{sel_upper_quantile_1617}
\scriptsize{
\begin{tabular}{c|cccccccc}
\hline\hline
& \multicolumn{2}{|c}{$\ell _{0}$-PQR} & $\ell _{0}$-CQR & $\ell _{1}$-PQR & 
AL-SCAD & AL-MCP & QR-SCAD & QR-MCP \\ 
~ & MIO & FO & ~ & ~ & ~ & ~ & ~ & ~ \\ \hline
1st & intercept & intercept & intercept & intercept & intercept & intercept & intercept & intercept \\ 
        ~ & (1,4.23) & (1,4.22) & (1,4.22) & (1,3.27) & (1,3.38) & (1,3.37) & (1,3.41) & (1,3.43) \\ \hline
        2nd & ~ & ~ & married & msmoke & medu3*B2(mage) & medu3*B2(mage) & medu3*B2(mage) & medu3*B2(mage) \\ 
        ~ & ~ & ~ & (0.2,0.03) & (0.5,-0.06) & (0.4,-10) & (0.4,-10) & (0.4,-10) & (0.4,-10) \\ \hline
        3rd & ~ & ~ & deadkids*B15(mage) & married*B16(npre) & fedu3*B1(mage) & fedu3*B1(mage) & fedu3*B1(mage) & fedu3*B1(mage) \\ 
        ~ & ~ & ~ & (0.2,-0.02) & (0.5,0.04) & (0.3,-6.41) & (0.3,-6.45) & (0.3,-6.62) & (0.3,-6.46) \\ \hline
        4th & ~ & ~ & married*B2(fage) & fhisp*B15(mage) & fedu3*B2(fage) & fedu3*B2(fage) & fedu3*B2(fage) & fedu3*B2(fage) \\ 
        ~ & ~ & ~ & (0.1,-0.01) & (0.5,-1.99) & (0.3,-10) & (0.3,-10) & (0.3,-10) & (0.3,-10) \\ \hline
        5th & ~ & ~ & married*B13(npre) & married*B17(npre) & B9(mage) & B9(mage) & married*B16(npre) & mhisp*B19(mage) \\ 
        ~ & ~ & ~ & (0.1,0.02) & (0.4,0) & (0.2,0.04) & (0.2,0.05) & (0.2,0.07) & (0.2,-10) \\ \hline
        6th & ~ & ~ & mhisp*B8(mage) & mhisp*B13(npre) & married*B15(mage) & married*B15(mage) & mhisp*B19(mage) & mhisp*B13(npre) \\ 
        ~ & ~ & ~ & (0.1,-0.01) & (0.4,-4.99) & (0.2,-0.06) & (0.2,-0.06) & (0.2,-10) & (0.2,-10) \\ \hline
        7th & ~ & ~ & mhisp*B15(npre) & fedu2*B9(mslb) & married*B9(mslb) & married*B16(npre) & mhisp*B13(npre) & fhisp*B14(mage) \\ 
        ~ & ~ & ~ & (0.1,0) & (0.4,0.01) & (0.2,0.05) & (0.2,0.04) & (0.2,-10) & (0.2,-10) \\ \hline
        8th & ~ & ~ & alcohol*B10(mage) & fedu3*B19(mslb) & mhisp*B19(mage) & married*B9(mslb) & fhisp*B14(mage) & deadkids*B9(mage) \\ 
        ~ & ~ & ~ & (0.1,0) & (0.4,-0.28) & (0.2,-10) & (0.2,0.05) & (0.2,-10) & (0.2,0.02) \\ \hline
        9th & ~ & ~ & alcohol*B11(mslb) & medu3*B2(mage) & mhisp*B13(npre) & married*B12(mslb) & msmoke*B19(mage) & deadkids*B14(fage) \\ 
        ~ & ~ & ~ & (0.1,0.03) & (0.4,-10) & (0.2,-10) & (0.2,0.04) & (0.2,-5.06) & (0.2,0.04) \\ \hline
        10th & ~ & ~ & alcohol*B14(mslb) & married*B5(mage) & fhisp*B14(mage) & mhisp*B19(mage) & fbaby*B14(mage) & msmoke*B19(mage) \\ 
        ~ & ~ & ~ & (0.1,-0.02) & (0.3,0) & (0.2,-10) & (0.2,-10) & (0.2,0.05) & (0.2,-5.12) \\ \hline
\end{tabular}
}
\begin{tablenotes} \item[]
For each parenthesized pair of values, the left value shows the proportion of the variable being selected across the sample splitting replications in the conformalized quantile regression procedure. The right value shows the corresponding averaged estimated regression coefficient value over those cases where the variable has been selected.
\end{tablenotes}
\end{threeparttable}
\end{landscape}

{\singlespacing
\bibliographystyle{econometrica}
\bibliography{L0QR_ref}
}

\end{document}